\documentclass{llncs}

\usepackage{bm}
\usepackage{pgfplots}
\usepackage{enumerate,paralist}
\usepackage{pbox}
\usepackage{makeidx}

\usepackage{appendix}
\usepackage{comment}
\usepackage{parskip}
\usepackage{multirow}
\usepackage{graphicx} 

\usepackage{tikz,pgffor}
\usetikzlibrary{arrows}
\usetikzlibrary{shapes}
\usetikzlibrary{calc}
\usetikzlibrary{automata}
\usetikzlibrary{positioning}

\tikzstyle{proglabel}=[shape=circle,draw,inner sep=0pt,minimum size=5mm]
\tikzstyle{tran}=[draw,->,>=stealth, rounded corners]

\usepackage[ruled, linesnumbered]{algorithm2e}
\usepackage{calc}
\usepackage{fp}

\usepackage{lmodern}
\usepackage{amsmath}
\usepackage{amssymb}

\usepackage{mathtools}
\usepackage{stmaryrd}

\usepackage{url}                  

\usepackage{listings}
\lstdefinelanguage{affprob}
{
morekeywords={prob, if, then, else, fi, while, do, od, true, false, and, or, skip},
sensitive = false
}

\DeclareMathAlphabet{\mathpzc}{OT1}{pzc}{m}{it}

\newcommand{\Nats}{\mathbb{N}}

\newcommand{\expv}{\mathbb{E}}
\newcommand{\Eval}{\mathsf{ET}}

\newcommand{\supp}{\mathrm{Supp}}

\newcommand{\wh}{\widehat}

\newcommand{\pvars}{X}
\newcommand{\rvars}{R}
\newcommand{\locs}{\mathit{L}}

\newcommand{\loc}{\ell}

\newcommand{\initval}{\vec{x}_0}

\newcommand{\Rset}{\mathbb{R}}
\newcommand{\Nset}{\mathbb{N}}
\newcommand{\Zset}{\mathbb{Z}}

\newcommand{\lin}{\mathit{in}}
\newcommand{\lout}{\mathit{out}}

\newcommand{\transitions}{\mapsto}

\newcommand{\assgn}[2]{[#1/#2]}
\newcommand{\id}{\mbox{\sl id}}
\newcommand{\probm}{\mathbb{P}}
\newcommand{\inv}{I}

\newcommand{\PRSMSynth}{{\sc PRSMSynth}}

\newcommand{\UB}{\mathsf{UB}}

\newcommand{\POLYS}{{\mathfrak{R}}{\left[x_1,\dots, x_{|X|}\right]}}
\newcommand{\POLYNS}[1]{{\mathfrak{R}}{\left[#1\right]}}
\newcommand{\DNF}[1]{\mathrm{DNF}(#1)}

\newcommand{\TRUE}{\mbox{\sl true}}
\newcommand{\FALSE}{\mbox{\sl false}}
\newcommand{\TERM}[1]{\mathsf{PP}(#1)}
\newcommand{\SAT}[1]{\left\llbracket{#1}\right\rrbracket}
\newcommand{\PO}[1]{\mbox{\sl PO}(#1)}

\newcommand{\support}[1]{\mathrm{supp}_{#1}}

\begin{document}

\title{Termination Analysis of Probabilistic Programs \\through Positivstellensatz's\thanks{A conference version of the paper appears in~\cite{CFK16}}}
\author{Krishnendu Chatterjee\inst{1} \and Hongfei Fu\inst{1,2} 
\and Amir Kafshdar Goharshady\inst{1}}
\institute{IST Austria \and 
State Key Laboratory of Computer
Science, Institute of Software, Chinese
Academy of Sciences, P.R. China
}

\maketitle

\begin{abstract}
We consider nondeterministic probabilistic programs with the most basic 
liveness property of termination.
We present efficient methods for termination analysis of nondeterministic 
probabilistic programs with polynomial guards and assignments. 
Our approach is through synthesis of polynomial ranking supermartingales,
that on one hand significantly generalizes linear ranking supermartingales 
and on the other hand is a counterpart of polynomial ranking-functions 
for proving termination of nonprobabilistic programs. 
The approach synthesizes  polynomial ranking-supermartingales 
through Positivstellensatz's,  
yielding an efficient method which is not only sound, but also
semi-complete over a large subclass of programs. 
We show experimental results to demonstrate that our approach can handle 
several classical programs with complex polynomial guards and assignments, 
and can synthesize efficient quadratic ranking-supermartingales when a linear 
one does not exist even for simple affine programs. 
\end{abstract}

\section{Introduction}\label{sec:introduction}

\noindent{\em Probabilistic Programs.}
Classic imperative programs extended with \emph{random-value generators} 
gives rise to probabilistic programs.
Probabilistic programs provide the appropriate framework to model
applications ranging from  randomized algorithms~\cite{RandBook,RandBook2}, 
to stochastic network protocols~\cite{BaierBook,prism}, 
to robot planning~\cite{KGFP09,kaelbling1998planning}, etc. 
Nondeterminism plays a crucial role in modeling, such as, to model
behaviors over which there is no control, or for abstraction. 
Thus nondeterministic probabilistic programs are crucial in a huge
range of problems, and hence their formal analysis has been studied 
across disciplines, such as probability theory and
statistics~\cite{Durrett,Howard,Kemeny,Rabin63,PazBook},
formal methods~\cite{BaierBook,prism},
artificial intelligence~\cite{LearningSurvey,kaelbling1998planning},
and programming languages~\cite{SriramCAV,HolgerPOPL,SumitPLDI,EGK12}.

\smallskip\noindent{\em Basic Termination Questions.}
Besides safety properties, the most basic property for analysis of 
programs is the liveness property. 
The most basic and widely used notion of liveness for programs is 
\emph{termination}.
In absence of probability (i.e., for nonprobabilistic programs), 
the synthesis of \emph{ranking functions} and proof of termination 
are equivalent~\cite{rwfloyd1967programs}, and numerous approaches
exist for synthesis of ranking functions for nonprobabilistic 
programs~\cite{DBLP:conf/cav/BradleyMS05,DBLP:conf/tacas/ColonS01,DBLP:conf/vmcai/PodelskiR04,DBLP:conf/pods/SohnG91}.
The most basic extension of the termination question for probabilistic 
programs is the \emph{almost-sure termination} question
which asks whether a program terminates with probability~1.
Another fundamental question is about \emph{finite termination} (aka positive almost-sure
termination~\cite{HolgerPOPL,BG05}) which asks whether the expected
termination time is finite.
The next interesting question is the \emph{concentration} bound computation 
problem that asks to compute a bound $M$ such that the probability that 
the termination time is below $M$ is concentrated,
or in other words, the probability that the termination time exceeds the
bound $M$ decreases exponentially.

\smallskip\noindent{\em Previous Results.}
We discuss the relevant previous results for termination analysis of probabilistic 
programs.

\begin{compactitem}
\item \emph{Probabilistic Programs.}
First, quantitative invariants was introduced to establish termination of
discrete probabilistic programs with demonic nondeterminism~\cite{MM04,MM05}, 
This was extended in~\cite{SriramCAV} to \emph{ranking  supermartingales} 
resulting in a sound (but not complete) approach to prove almost-sure 
termination of probabilistic programs without nondeterminism but with 
integer- and real-valued random variables from distributions like uniform, 
Gaussian, and Poison, etc.
For probabilistic programs with countable state-space and without
nondeterminism, the Lyapunov ranking functions provide a sound and complete
method for proving finite termination~\cite{BG05,Foster53}.
Another sound method is to explore bounded-termination with exponential 
decrease of probabilities~\cite{DBLP:conf/sas/Monniaux01} 
through abstract interpretation~\cite{DBLP:conf/popl/CousotC77}. 
For probabilistic programs with nondeterminism, a sound and complete 
characterization for finite termination through ranking-supermartingale 
is obtained in~\cite{HolgerPOPL}.
Ranking supermartingales thus provide a very powerful approach for
termination analysis of probabilistic programs. 

\item \emph{Ranking Functions/supermartingales Synthesis.}
Synthesis of linear ranking-functions/ranking-supermartingales has been studied 
extensively in~\cite{DBLP:conf/vmcai/PodelskiR04,DBLP:conf/tacas/ColonS01,SriramCAV,DBLP:conf/popl/ChatterjeeFNH16}.
In context of probabilistic programs, the algorithmic study of synthesis of 
linear ranking supermartingales for probabilistic programs (cf.~\cite{SriramCAV}) 
and probabilistic programs with nondeterminism (cf. our previous result~\cite{DBLP:conf/popl/ChatterjeeFNH16}) has been studied.
The major technique adopted in these results is Farkas' Lemma~\cite{FarkasLemma} 
which serves as a complete reasoning method for linear inequalities. 
Beyond linear ranking functions, polynomial ranking functions have also been considered.
Heuristic synthesis method of polynomial ranking-functions is studied 
in~\cite{DBLP:journals/fac/BabicCHR13,DBLP:conf/vmcai/BradleyMS05}: 
Cook~\emph{et al.}~\cite{DBLP:journals/fac/BabicCHR13} checked termination of 
deterministic polynomial programs by detecting divergence on program variables 
and Bradley~\emph{et al.}~\cite{DBLP:conf/vmcai/BradleyMS05} extended to nondeterministic 
programs through an analysis on finite differences over transitions. 
More general methods for deterministic polynomial programs are given 
by~\cite{DBLP:conf/vmcai/Cousot05,DBLP:journals/jossac/ShenWYZ13} 
where Cousot~\cite{DBLP:conf/vmcai/Cousot05} uses Lagrangian Relaxation, 
and Shen~\emph{et al.}~\cite{DBLP:journals/jossac/ShenWYZ13} use 
Putinar's Positivstellensatz~\cite{PutinarPositivstellensatz}.  
Complete methods of synthesizing polynomial ranking-functions for nondeterministic programs 
are studied by Yang~\emph{et al.}~\cite{DBLP:journals/fcsc/YangZZX10}, where 
a complete method through root classification/real root isolation of semi-algbebraic 
systems and quantifier elimination is proposed.  

\end{compactitem}
To summarize, while many different approaches has been studied, the algorithmic 
study of synthesis of ranking supermartingales for probabilistic programs has only
been limited to linear ranking supermartingales. 
For example,~\cite{SriramCAV} presents a method of synthesis of linear ranking
supermartingales for probabilistic programs without nondeterminism, and identifies
synthesis of more general nonlinear supermartingales, or extension to probabilistic
programs with nondeterminism as important challenges. While the approach of~\cite{SriramCAV}
has been extended to probabilistic programs with nondeterminism in our previous result~\cite{DBLP:conf/popl/ChatterjeeFNH16}, 
it is restricted to linear ranking supermartingales.
Hence there is no algorithmic approach to handle nonlinear ranking supermartingales
even for probabilistic programs without nondeterminism.

\smallskip\noindent{\em Our Contributions.} Our contributions are as follows:
\begin{compactenum}

\item \emph{Polynomial Ranking Supermartingales.} 
First, we extend the notion of linear ranking supermartingales (LRSM)
to polynomial ranking supermartingales (pRSM).
We show (by a straightforward extension of LRSM) that pRSM implies both 
almost-sure as well as finite termination.

\item \emph{Positivstellensatz's.}
Second, we conduct a detailed investigation on the application of 
Positivstellensatz's (German for ``positive-locus-theorem'' which 
is related to polynomials over semialgebraic sets)
(cf. Sect.~\ref{sect:positivstellensatz}) to synthesis of 
pRSMs over nondeterministic probabilistic programs. 
To the best of our knowledge, this is the first result which demonstrates the 
synthesis of a polynomial subclass of ranking supermartingales through Positivstellensatz's.

\item \emph{New Approach for Non-probabilistic Programs.}
Our results also extend existing results 
for nonprobabilistic programs. 
We present the first result that uses Schm\"{u}dgen's Positivstellensatz~\cite{SchmudgenPositivstellensatz} 
and Handelman's Theorem~\cite{HandelmanTheorem} to synthesize polynomial 
ranking-functions for nonprobabilistic programs.

\item \emph{Efficient Approach.} 
The previous complete method~\cite{DBLP:journals/fcsc/YangZZX10} 
suffers from high computational complexity due to the use of quantifier elimination.
In contrast, our approach (sound but not complete) is efficient since the synthesis can be 
accomplished through linear or semi-definite programming, 
which can mostly be solved in polynomial time in the problem size~\cite{SemidefiniteProgramming}. 
In particular, our approach does not require quantifier elimination, and works 
for nondeterministic probabilistic programs.

\item \emph{Experimental Results.}
We demonstrate the effectiveness of our approach on several classical examples.
We show that on classical examples, such as Gambler's Ruin, and Random Walk, 
our approach can synthesize a pRSM efficiently. 
For these examples, LRSMs do not exist, and many of them cannot be analysed efficiently by previous approaches. 

\end{compactenum}

\smallskip\noindent{\em Technical Contributions and Novelty.}
The main technical contributions and novelty are: 
\begin{compactenum}

\item While Farkas' Lemma and Motzkin's Transposition Theorem
are standard techniques to linear ranking functions or linear
ranking supermartingales, they are not sufficient for pRSMs. 
Instead, our technical contributions is to use various Positivstellensatz's
to synthesize pRSMs.

\item Even for nonprobabilistic programs, only a limited number of 
Positivstellensatz's have been used, e.g.,~\cite{DBLP:journals/jossac/ShenWYZ13};
some of the Positivstellensatz's we use 
(such as Schm\"{u}dgen's Positivstellensatz~\cite{SchmudgenPositivstellensatz} 
and Handelman's Theorem~\cite{HandelmanTheorem}) have not even been used in 
the context of nonprobabilistic programs. 

\end{compactenum}
In summary, we study the use of Positivstellensatz's for the first time for probabilistic 
programs, and for some of them, even for the first time for nonprobabilistic 
programs, and show that how they can be used for efficient algorithms for program analysis. 

\smallskip\noindent{\em Organization of the Paper.} 
In Sect.~\ref{sect:probprog}, we present the syntax and semantics of probabilistic programs. 
In Sect.~\ref{sect:terminationquestion}, we define the problems to be studied.
In Sect.~\ref{sect:prsm}, we develop the notion of polynomial ranking supermartingale. 
Then in Sect.~\ref{sect:synthesisalgorithm}, we present Positivstellensatz and 
our algorithm to synthesize polynomial ranking supermartingales.
In Sect.~\ref{sect:experimental}, we present our experimental results. 
Finally, Sect.~\ref{sect:conclusion} concludes the paper.

\section{Probabilistic Programs}\label{sect:probprog}

\subsection{Basic Notations and Concepts}

For a set $A$, we denote by $|A|$ the cardinality of $A$. 
We denote by $\Nset$, $\Nset_0$, $\Zset$, and $\Rset$ the sets of all positive integers, non-negative integers, integers, and real numbers, respectively. 
We use boldface notation for vectors, e.g. $\vec{x}$, $\vec{y}$, etc, and we denote an $i$-th component of a vector $\vec{x}$ by $\vec{x}[i]$. 

\smallskip\noindent{\em Polynomial Predicates.} 
Let $X$ be a finite set of variables endowed with a fixed linear order under which we have $X=\{x_1,\dots,x_{|X|}\}$. 
We denote the set of real-coefficient polynomials by $\POLYS$ or $\POLYNS{X}$. 
A \emph{polynomial constraint} over $X$ is a logical formula of the form ${g_1}{\Join}{g_2}$, where $g_1,g_2$ are polynomials over $X$ and $\Join\in\{<,\le,>,\ge\}$. 
A \emph{propositional polynomial predicate} over $X$ is a propositional formula whose all atomic propositional literals are either $\TRUE,\FALSE$ or polynomial constraints over $X$.
The validity of the satisfaction assertion $\vec{x}\models\phi$ between a vector $\vec{x}\in\Rset^{|\pvars|}$ (interpreted in the way that the value for $x_j$ $(1\le j\le |X|)$ is $\vec{x}[j]$)
and a propositional polynomial predicate $\phi$ is defined in the standard way w.r.t polynomial evaluation and normal semantics for logical connectives (cf.  Appendix~\ref{sect:predicate}). 
The satisfaction set of a propositional polynomial predicate $\phi$ is defined as $\SAT{\phi}:=\{\vec{x}\in\Rset^{|\pvars|}\mid \vec{x}\models\phi\}$. 
For more on polynomials (e.g., polynomial evaluation and arithmetic over polynomials), we refer to the textbook~\cite[Chapter 3]{AlgebraHungerford}.

\smallskip\noindent{\em Probability Space.} A \emph{probability space} is a triple $(\Omega,\mathcal{F},\probm)$, where $\Omega$ is a non-empty set (so-called \emph{sample space}), $\mathcal{F}$ is a \emph{$\sigma$-algebra} over $\Omega$ (i.e., a collection of subsets of $\Omega$ that contains the empty set $\emptyset$ and is closed under complementation and countable union), and $\probm$ is a \emph{probability measure} on $\mathcal{F}$, i.e., a function $\probm\colon \mathcal{F}\rightarrow[0,1]$ such that (i) $\probm(\Omega)=1$ and 
(ii) for all set-sequences $A_1,A_2,\dots \in \mathcal{F}$ that are pairwise-disjoint
(i.e., $A_i \cap A_j = \emptyset$ whenever $i\ne j$) 
it holds that $\sum_{i=1}^{\infty}\probm(A_i)=\probm\left(\bigcup_{i=1}^{\infty} A_i\right)$~. 

\smallskip\noindent{\em Random Variables and Filtrations.} A \emph{random variable} $X$ in a probability space $(\Omega,\mathcal{F},\probm)$ is an $\mathcal{F}$-measurable function $X\colon \Omega \rightarrow \Rset \cup \{-\infty,+\infty\}$, i.e., 
a function satisfying the condition that for all $d\in \Rset \cup \{+\infty, -\infty\}$, the set $\{\omega\in \Omega\mid X(\omega)\leq d\}$ belongs to $\mathcal{F}$. 
The \emph{expected value} of a random variable $X$, denote by $\expv(X)$, is defined as the Lebesgue integral of $X$ with respect to $\probm$, i.e., 
$\expv(X):=\int X\,\mathrm{d}\probm$~;
the precise definition of Lebesgue integral is somewhat technical and is 
omitted  here (cf.~\cite[Chapter 5]{Billingsley:book} for a formal definition).
A \emph{filtration} of a probability space $(\Omega,\mathcal{F},\probm)$ is an infinite sequence $\{\mathcal{F}_n \}_{n\in\Nset_0}$ of $\sigma$-algebras over $\Omega$ such that $\mathcal{F}_n \subseteq \mathcal{F}_{n+1} \subseteq\mathcal{F}$ for all $n\in\Nset_0$. 

\subsection{Probabilistic Programs}

\noindent{\bf The Syntax.}
The class of probabilistic programs we consider encompasses basic programming mechanisms
such as assignment statement (indicated by `:='), while-loop, if-branch, basic probabilistic mechanisms such as probabilistic branch 
(indicated by `prob') and random sampling, 
and demonic nondeterminism indicated by `$\star$'.  
Variables (or identifiers) of a probabilistic program are of \emph{real} type, i.e., values of the variables are real numbers; moreover, 
variables are classified into \emph{program} and \emph{sampling} variables, where program variables receive their values through assignment statements and sampling variables do through random samplings. 
We consider that each sampling variable $r$ is {\em bounded}, i.e., associated 
with a one-dimensional cumulative distribution function $\Upsilon_r$ and 
a non-empty bounded interval $\support{r}$ such that any random variable $z$ 
which respects $\Upsilon_r$ satisfies that $z$ lies in the bounded interval 
with probability~1. 
Due to space limit, we put details (e.g., grammar) in Appendix~\ref{app:syntax}. 
An example probabilistic program is illustrated in Example~\ref{ex:runningexample}. 

\lstset{language=affprob}
\lstset{tabsize=3}
\newsavebox{\progrunningexample}
\begin{lrbox}{\progrunningexample}
\begin{lstlisting}[mathescape]

1: while $1\le x\wedge x\le 10$ do
2:   if $\star$ then 
3:      $x:=x+r$  
     else 
4:     if prob($0.51$) then 
5:        $x:=x-1$ 
       else 
6:        $x:=x+1$ 
       fi
     fi
   od
7:

\end{lstlisting}
\end{lrbox}

\begin{example}\label{ex:runningexample}
Consider the running example depicted in Fig.~\ref{fig:runningexample}, where $r$
is a sampling  variable with the two-point distribution $\{1\mapsto 0.5,-1\mapsto 0.5\}$ where the probability to take values $1$ and $-1$ are both $0.5$. 
The probabilistic program models a scenario of Gambler's Ruin where the gambler has initial money $x$ and repeats gambling until he wins more than $10$ or lose all his money. 
The result of a gamble is nondeterministic: either win $1$ with probability $0.5$ (nondeterministic branch); 
or lose with probability $0.51$ (the probabilistic branch).
The numbers $1-7$ on the left are the program counters for the program, where $1$ is the initial program counter and $7$ the terminal program counter. 
\end{example}

\begin{figure}
\begin{minipage}{0.5\textwidth}
\centering
\usebox{\progrunningexample}
\caption{The Running Example: Gambler's Ruin}
\label{fig:runningexample}
\end{minipage}
\begin{minipage}{0.5\textwidth}
\centering
\begin{tikzpicture}[x = 2cm]

\node[proglabel] (while) at (1.5,-1.5) {$1$};
\node[proglabel] (demon) at (3,-1.5) {$2$};
\node[proglabel] (probA) at (2.5,0)    {$3$};
\node[proglabel] (probB) at (2.5,-3.5)   {$4$};
\node[proglabel] (assC)  at (1.8,-3.5)   {$5$};
\node[proglabel] (assD)  at (2.5,-2)   {$6$};
\node[proglabel] (fin)   at (1.2,0)   {$7$};

\draw[tran] (while) to node[auto, font=\scriptsize] {$x<1\vee x>10$} (fin);
\draw[tran] (while) to node[auto, font=\scriptsize] {$x\geq 1\wedge x\le 10$} (demon);

\draw[tran] (probA) --  node[auto, swap, pos=0.7] {$x\mapsto x+r$} (probA-|while) -- (while);
\draw[tran] (probB) -- node[auto, font=\scriptsize] {$0.51$}  (assC);
\draw[tran] (probB) -- node[auto, font=\scriptsize] {$0.49$}  (assD);

\draw[tran] (demon) -- node[auto, font=\scriptsize] {$\star$} (demon|-probA) -- (probA);
\draw[tran] (demon) -- node[auto, font=\scriptsize] {$\star$} (demon|-probB) -- (probB);
\draw[tran] (assC) --  (assC-|while) -- node[auto, pos=0.2] {$x\mapsto x-1$} (while);
\draw[tran] (assD) -- node[auto] {$x\mapsto x+1$} (assD-|while.305) --  (while.-55);
\end{tikzpicture}
\caption{The CFG of the Running Example}
\label{fig:runningcfg}
\end{minipage}
\end{figure}


\smallskip\noindent{\bf The semantics.} We use control flow graphs to capture the semantics
of probabilistic programs, which we define below.

\begin{definition}[Control Flow Graph]
\label{def:CFG}
A \emph{control flow graph} (CFG) is a tuple $\mathcal{G}=(\locs,\bot,(\pvars,\rvars),\transitions)$ with the following components:
\begin{compactitem}
\item $\locs$ is a finite set of \emph{labels} partitioned into four pairwise-disjoint subsets $\locs_\mathrm{d}$, $\locs_\mathrm{p},\locs_\mathrm{c}$ and $\locs_\mathrm{a}$ 
of demonic, probabilistic, conditional-branching (branching for short) and assignment labels, resp.; 
and $\bot$ is a special label not in $L$ called the \emph{terminal label};
\item $\pvars$ and $\rvars$ are disjoint finite sets of real-valued \emph{program} and \emph{sampling variables}  respectively ;
\item $\transitions$ is a \emph{transition relation} in which every member (called \emph{transition}) is a tuple of the form 
$(\loc,\alpha,\loc')$ for which $\loc$ (resp. $\loc'$) is the \emph{source label} (resp. \emph{target label}) in $\locs$ 
and $\alpha$ is either a real number in $(0,1)$ if $\loc\in\locs_\mathrm{p}$, or $\star$ if $\loc\in\locs_\mathrm{d}$, 
or a propositional polynomial predicate if $\loc\in\locs_\mathrm{c}$, or an \emph{update function} $f\colon \Rset^{|\pvars|}\times\Rset^{|\rvars|}\rightarrow \Rset^{|\pvars|}$ if $\loc\in\locs_\mathrm{a}$. 
\end{compactitem}
\end{definition}
W.l.o.g, we assume that $\locs\subseteq\Nset_0$. 
Intuitively, labels in $\locs_\mathrm{d}$ correspond to demonic statements indicated by `$\star$'; labels in $\locs_\mathrm{p}$ correspond to 
probabilistic-branching statements indicated by `\textbf{prob}'; 
labels in $\locs_\mathrm{c}$ correspond to conditional-branching statements indicated by some propositional polynomial predicate;
labels in $\locs_\mathrm{a}$ correspond to assignments indicated by `$:=$' and the terminal label $\bot$ denotes the termination of a program. 
The transition relation $\transitions$ specifies the transitions between labels together with the additional information specific to different types of labels. 
The update functions are interpreted as follows: we first fix two linear orders  
on $\pvars$ and $\rvars$ so that $\pvars = \{x_1,\dots,x_{|\pvars|}\}$ and  $\rvars = \{r_1,\dots,r_{|\rvars|}\}$, interpreting each vector $\vec{x}\in\Rset^{|\pvars|}$ (resp. $\vec{r}\in\Rset^{|\rvars|}$) as a \emph{valuation} of program (resp. sampling) variables in the sense that the value of $x_j$ (resp. $r_j$) is $\vec{x}[j]$ (resp. $\vec{r}[j]$);
then each update function $f$ is interpreted as a function which transforms a valuation $\vec{x}\in\Rset^{|\pvars|}$ before the execution of an assignment statement into 
$f(\vec{x},\vec{r})$ after the execution of the assignment statement, where 
$\vec{r}$ is the valuation on $\rvars$ obtained from a sampling before the execution of 
the assignment statement. 

It is intuitively clear that any probabilistic program can be naturally transformed into a CFG.
Informally, each label represents a program location in an execution of a probabilistic program 
for which the statement of the program location is the next to be executed. 
A detailed construction is provided in Appendix~\ref{app:cfg}. 

\begin{example}\label{ex:runningcfg}
The control flow graph of the running example (Example~\ref{ex:runningexample}) is depicted in Fig.~\ref{fig:runningcfg}, where
vertices correspond to labels specified in Fig.~\ref{fig:runningexample}. 
\end{example}

Now we present the semantics of probabilistic programs. 
In the rest of the section, we fix a probabilistic program $P$ with the set $\pvars = \{x_1,\dots,x_{|\pvars|}\}$ of program variables and the set $\rvars = \{r_1,\dots,r_{|\rvars|}\}$ of sampling variables, 
and let $\mathcal{G}=(\locs,\bot,(\pvars,\rvars),\transitions)$ be its associated CFG. 
We also fix $\loc_0$ and resp. $\initval$ to be the label corresponding to the first statement to be executed in $P$ and resp. the initial valuation of program variables.

\paragraph*{The Semantics.} A \emph{configuration} (for $P$) is a tuple $(\loc,\vec{x})$ where $\loc\in\locs\cup\{\bot\}$ and $\vec{x}\in\Rset^{|\pvars|}$. 
A \emph{finite path} (of $P$) is a finite sequence of configurations $(\loc_0,\initval),\cdots,(\loc_k,\vec{x}_k)$ such that for all $0 \leq i < k$,
either (i) $\loc_{i+1}=\loc_i=\bot$ and $\vec{x}_i=\vec{x}_{i+1}$ (i.e., the program terminates); 
or (ii)~there exist $(\loc_i,\alpha,\loc_{i+1})\in\transitions$ and $\vec{r}\in\{\vec{r}'\mid\forall r\in\rvars.\ \vec{r}'(r)\in \support{r}\}$ such that
one of the following conditions hold: 
(a)~$\loc_i\in \locs_\mathrm{p}\cup\locs_\mathrm{d}$ and $\vec{x}_i=\vec{x}_{i+1}$ (probabilistic or demonic transitions),
(b)~$\loc_i\in \locs_\mathrm{c}$, $\vec{x}_i=\vec{x}_{i+1}$ and $\vec{x}_i\models\alpha$ (conditional-branch transitions),
(c)~$\loc_i\in \locs_\mathrm{a}$ and $\vec{x}_{i+1}=\alpha(\vec{x}_i,\vec{r})$ (assignment transitions).
A \emph{run} (of $P$) is an infinite sequence of configurations whose all finite prefixes are finite paths over $P$.
A configuration $(\loc,\vec{x})$ is {\em reachable} from the initial configuration $(\loc_0,\initval)$ if there exists a finite path $(\loc_0,\initval),\cdots,(\loc_k,\vec{x}_k)$ such that $(\loc,\vec{x})=(\loc_k,\vec{x}_k)$.

The probabilistic feature of $P$ can be captured by constructing a suitable probability measure over the set of all its runs. 
However, before this can be done, nondeterminism in $P$ needs to be resolved by some \emph{scheduler}. 

\begin{definition}[Scheduler]\label{def:scheduler}
A \emph{scheduler} (for $P$) is a function which assigns to every finite path $(\loc_0,\initval),\dots,(\loc_k,\vec{x}_k)$ with $\loc_k\in \locs_\mathrm{d}$ a transition in $\transitions$ with source label $\loc_k$. 
\end{definition}

The behaviour of $P$ under a scheduler $\sigma$ is standard: at each step, $P$ first samples a real number for each sampling variable and then evolves to the next step according to its CFG or the scheduler choice
(the details are in Appendix~\ref{app:semanticsdetails}). 
In this way, the scheduler and random choices/samplings produce a run 
over $P$. 
Moreover, each scheduler $\sigma$ induces a unique probability measure 
$\probm^{\sigma}$ over the runs of $P$. 
In the rest of the paper, we will use $\expv^{\sigma}(\cdot)$ to denote the 
expected values of random variables under $\probm^{\sigma}$. 

\smallskip\noindent{\em Random Variables and Filtrations over Runs.} 
We define the following (vectors of) random variables on the set of runs of $P$: 
$\{\theta^P_n\}_{n\in\Nset_0},~\{\overline{\vec{x}}^P_{n}\}_{n\in\Nset_0}$ and $\{\overline{\vec{r}}^P_{n}\}_{n\in\Nset_0}$~:
each $\theta^P_n$ is the random variable representing the (integer-valued) label at the $n$-th step;
each $\overline{\vec{x}}^P_{n}$ is the vector of random variables such that each $\overline{\vec{x}}^P_{n}[i]$ is the random variable representing the value of the program variable $x_i$ at the $n$-th step;
and each $\overline{\vec{r}}^P_{n}[i]$ is the random variable representing the sampled value of the sampling variable $r_i$ at the $n$-th step. 
The filtration $\{\mathcal{H}^P_n\}_{n\in\mathbb{N}_0}$ is defined such that 
each $\sigma$-algebra $\mathcal{H}^P_n$ is the smallest $\sigma$-algebra that makes all random variables in $\{\theta^P_k\}_{0\le k\le n}$ and $\{\overline{\vec{x}}^P_{k}\}_{0\le k\le n}$ measurable. 
We will omit the superscript $P$ in all the notations above if it is clear from the context. 

\begin{remark}
Under the condition that each sampling variable is bounded, using an inductive argument it follows that each $\overline{\vec{x}}_{n}$ is a vector of bounded random variables. 
Thus $\expv^\sigma({|}{\overline{\vec{x}}_n[i]}{|})$ exists for each random variable $\overline{\vec{x}}_n[i]$.   
\end{remark}

Below we define the notion of \emph{polynomial invariants} which logically captures all reachable configurations. 
A polynomial invariant may be obtained through abstract interpretation~\cite{DBLP:conf/popl/CousotC77}. 

\begin{definition}[Polynomial Invariant]
A \emph{polynomial invariant} (for $P$) is a function $\inv$ assigning a propositional polynomial predicate over $\pvars$ to every label in $\mathcal{G}$ such that for all configurations $(\loc,\vec{x})$ reachable from $(\loc_0,\initval)$ in $\mathcal{G}$, it holds that $\vec{x}\models \inv(\loc)$.
\end{definition}

\section{Termination over Probabilistic Programs}\label{sect:terminationquestion}
In this section, we first define the notions of almost-sure/finite termination 
and concentration bounds over probabilistic programs, and then describe the computational problems 
studied in this paper.  
Below we fix a probabilistic program $P$ with its associated CFG $\mathcal{G}=(\locs,\bot,(\pvars,\rvars),\transitions)$ 
and an initial configuration $(\loc_0,\initval)$ for $P$. 

\begin{definition}[Termination~\cite{BG05,HolgerPOPL,DBLP:conf/popl/ChatterjeeFNH16}]
A run $\omega=\{(\loc_n,\vec{x}_n)\}_{n\in\Nset_0}$ over $P$ is \emph{terminating} if  $\loc_n=\bot$ for some $n\in\Nset_0$.
The \emph{termination time} of $P$ is a random variable $T_P$ such that for each run $\omega=\{(\loc_n,\vec{x}_n)\}_{n\in\Nset_0}$, $T_P(\omega)$ is the least number $n$ such that $\loc_n=\bot$ 
if such $n$ exists, and $\infty$ otherwise. 
The program $P$ is said to be \emph{almost-sure terminating} (resp. \emph{finitely terminating}) if $\probm^\sigma(T_P<\infty)=1$ 
(resp. $\expv^\sigma(T_P)<\infty$) for all schedulers $\sigma$ (for $P$). 
\end{definition}
Note that $\expv^\sigma(T_P)<\infty$ implies that 
$\probm^\sigma(T_P<\infty)=1$, but the converse does not necessarily hold
(see~\cite[Example~5]{SriramCAV} for an example).
i.e., finite-termination implies almost-sure termination, but not vice-versa.  
To measure the expected values of the termination time under all (demonic) schedulers, we further define the quantity $\Eval(P):=\sup_{\sigma}\expv^{\sigma}(T_P)$~.

\begin{definition}[Concentration on Termination Time~\cite{DBLP:conf/sas/Monniaux01,DBLP:conf/popl/ChatterjeeFNH16}]
A \emph{concentration bound} for $P$ is a non-negative integer $M$ such that there exist real constants 
$c_1\ge 0$ and $c_2>0$, and for all $N \geq M$ we have 
$\probm(T_P>N)\le c_1\cdot e^{-c_2\cdot N}$ . 
\end{definition}
Informally, a concentration bound characterizes exponential decrease of probability values of non-termination beyond the bound. 
On one hand, it can be used to give an upper bound on probability of non-termination beyond a large step; 
and on the other hand, it leads to an algorithm that approximates  $\Eval(P)$ (cf.~\cite[Theorem 5]{DBLP:conf/popl/ChatterjeeFNH16}). 

In this paper, we consider the algorithmic analysis of the following problems: 
\begin{compactitem}
\item \textbf{Input:} a probabilistic program $P$, a polynomial invariant $\inv$ for $P$ and an initial configuration $(\loc_0,\initval)$ for $P$;
\item \textbf{Output: (Almost-Sure/Finite Termination)} ``$\mbox{\sl yes}$'' if the algorithm finds that $P$ is almost-sure/finite terminating and ``$\mbox{\sl fail}$'' otherwise;
\item \textbf{Output: (Concentration on Termination)} a concentration bound if the algorithm finds one and  ``$\mbox{\sl fail}$'' otherwise. 
\end{compactitem}

\section{Polynomial Ranking-Supermartingale}\label{sect:prsm}

In this section, we develop the notion of polynomial ranking-supermartingale which is an extension of linear ranking-supermartingale~\cite{SriramCAV,DBLP:conf/popl/ChatterjeeFNH16}.
We fix a probabilistic program $P$, a polynomial invariant $I$ for $P$ and an initial configuration $(\loc_0,\initval)$ for $P$.
Let $\mathcal{G}=(\locs,\bot,(\pvars,\rvars),\transitions )$ be the associated CFG of $P$, with $\pvars = \{x_1,\dots,x_{|\pvars|}\}$ and $\rvars = \{r_1,\dots,r_{|\rvars|}\}$. 
We first present the general notion of \emph{ranking supermartingale}, and then define that of \emph{polynomial ranking supermartingale}.

\begin{definition}[Ranking Supermartingale~\cite{HolgerPOPL,DBLP:conf/popl/ChatterjeeFNH16}]\label{def:rsm}
A discrete-time stochastic process $\{X_n\}_{n\in\Nset_0}$ w.r.t a
filtration $\{\mathcal{F}_n\}_{n\in\mathbb{N}_0}$ is a \emph{ranking supermartingale} (RSM) if there exist $K<0$ and $\epsilon>0$ such that for all $n\in\mathbb{N}_0$, we have $\expv(|X_n|)<\infty$ and it holds almost surely (with probability~$1$) that
$X_n\ge K$\mbox{ and }$\expv(X_{n+1}\mid \mathcal{F}_n)\le X_n-\epsilon\cdot\mathbf{1}_{X_n\ge 0}$\enskip,
where $\expv(X_{n+1}\mid \mathcal{F}_n)$ is the conditional expectation of $X_{n+1}$ given  $\mathcal{F}_n$ (cf.~\cite[Chapter~9]{probabilitycambridge}).
\end{definition}

Informally, a polynomial ranking-supermartingale over $P$ is a polynomial instantiation of 
an RSM through certain function 
$\eta:(\locs\cup\{\bot\})\times \Rset^{|\pvars|}\rightarrow\Rset$ which satisfies that each $\eta(\loc,\cdot)$ (for all $\loc\in\locs\cup\{\bot\}$) is essentially a polynomial function over $\pvars$.
Given such a function $\eta$, the intuition is to have conditions that make 
the stochastic process $X_n=\eta(\theta_n,\overline{\vec{x}}_n)$ an RSM.
To ensure this, we consider the conditional expectation $\expv^\sigma\left(X_{n+1}\mid\mathcal{H}_n\right)$;
this is captured by an extension of \emph{pre-expectation}~\cite{SriramCAV,DBLP:conf/popl/ChatterjeeFNH16} from the linear to the polynomial case.
Below we define $\locs_{\bot}:=\locs\cup\{\bot\}$.
For a function $g:\Rset^{|\pvars|}\times\Rset^{|\rvars|}\rightarrow\Rset$, we let $\expv_R(g,\cdot):\Rset^{|\pvars|}\rightarrow\Rset$ be the function such 
that each $\expv_R(g,\vec{x})$ is the expected value $\expv(g(\vec{x},\hat{\vec{r}}))$,
where $\hat{\vec{r}}$ is any vector of independent random variables such that each $\hat{\vec{r}}[i]$ is a random variable that respects the cumulative 
distribution function $\Upsilon_{r_i}$.

\begin{definition}[Pre-Expectation]\label{def:preexpectation}
Let $\eta: \locs_\bot\times\mathbb{R}^{|\pvars|}\rightarrow\mathbb{R}$ be a function such that
each $\eta(\loc,\cdot)$ (for all $\loc\in\locs_\bot$) is a polynomial function over $\pvars$.
The function $\mathrm{pre}_\eta: \locs_\bot\times\mathbb{R}^{|\pvars|}\rightarrow\mathbb{R}$ is defined by:
\begin{compactitem}
\item $\mathrm{pre}_\eta(\loc,\vec{x}):=\sum_{(\loc,z,\loc')\in\transitions} z\cdot \eta\left(\loc',\vec{x}\right)$ if $\loc\in\locs_\mathrm{p}$ 
(probabilistic transitions);
\item $\mathrm{pre}_\eta(\loc,\vec{x}):=\max_{(\loc,\star,\loc')\in\transitions}\eta(\loc',\vec{x})$ if $\loc\in\locs_\mathrm{d}$ (nondeterministic transitions);
\item $\mathrm{pre}_\eta(\loc,\vec{x}):=\eta(\loc',\vec{x})$ if $\loc\in\locs_\mathrm{c}$ and $(\loc,\phi,\loc')$ is the only transition in $\transitions$ such that $\vec{x}\models\phi$ (conditional transitions);
\item $\mathrm{pre}_\eta(\loc,\vec{x}):=\expv_{\rvars}\left(g,\vec{x}\right)$
if $\loc\in\locs_{\mathrm{a}}$, where $g$ is the function such that 
$g(\vec{x},\vec{r})=\eta\left(\loc',f(\vec{x},\vec{r})\right)$ and $(\loc,f,\loc')$ 
is the only transition in $\transitions$ (assignment transitions); and 
\item $\mathrm{pre}_\eta(\loc,\vec{x}):=\eta(\loc,\vec{x})$ if $\loc=\bot$ (terminal location).
\end{compactitem}
\end{definition}
The following lemma establishes the relationship between pre-expectation and conditional expectation whose proof is in Appendix~\ref{sect:prsmproof}.

\begin{lemma}\label{lemma:condexp}
Let $\eta: \locs_\bot\times\mathbb{R}^{|\pvars|}\rightarrow\mathbb{R}$ be a function such that
each $\eta(\loc,\cdot)$ (for all $\loc\in\locs_\bot$) is a polynomial function over $\pvars$, and $\sigma$ be any scheduler.
Let the stochastic process $\{X_n\}_{n\in\mathbb{N}_0}$ be defined by:
$X_{n}:=\eta(\theta_{n},\overline{\vec{x}}_{n})$.
Then for all $n\in\mathbb{N}_0$, we have $\expv^{\sigma}(X_{n+1}\mid\mathcal{H}_n)\le\mathrm{pre}_\eta(\theta_{n},\overline{\vec{x}}_{n})$.
\end{lemma}

\begin{example}\label{ex:preexpectation}
Consider the running example in Example~\ref{ex:runningexample} with CFG in Fig.~\ref{fig:runningcfg}. 
Let $\eta$ be the function specified in the second and fifth column of Table~\ref{tab:runningexample}, where $g(x):=(x-1)(10-x)$.
Then $\mathrm{pre}_\eta$ is given in the third and sixth column of Table~\ref{tab:runningexample}. 
Note that the case for $i=2$ is obtained from $\mathrm{pre}_\eta(2, x)=\max\{g(x)+9.6,g(x)+9.6\}$, and 
the case for $i=3$ is from $\mathrm{pre}_\eta(3, x)=\expv_R(h , x)$, where $h$ is the function
$h(y,r)= g(y)-(2y-11)r-r^2+10$.
\end{example}

\begin{table}
\begin{center}
\begin{tabular}{|c|c|c|c|c|c|}
\hline
\ \ $i$ \ \ & $\eta(i,x)$ & $\mathrm{pre}_\eta(i,x)$ & \ \ $i$ \ \ & $\eta(i,x)$ & $\mathrm{pre}_\eta(i,x)$    \\
\hline
\multirow{2}{*}{$1$} & \multirow{2}{*}{$g(x)+10$} & \ \ $\mathbf{1}_{1\le x\le 10}\cdot (g(x)+9.8)$ \ \  &  \multirow{2}{*}{$5$}  & \multirow{2}{*}{$g(x)+2x-1.8$} & \multirow{2}{*}{$g(x)+2x-2$} \\
  &  &   ${}+\mathbf{1}_{x < 1\vee x> 10}\cdot (-0.2)$ & & & \\
\hline
$2$ & \ \ $g(x)+9.8$ \ \ & $g(x)+9.6$ & $6$  & \ \ $g(x)-2x+20.2$ \ \ & \ \ $g(x)-2x+20$ \ \  \\
\hline
$3$ & $g(x)+9.6$ & $g(x)+9$ & $7$  & $-0.2$ & $-0.2$ \\
\hline
$4$  &  $g(x)+9.6$  &   $g(x)+0.04x+8.98$  & & & \\
\hline                       
\end{tabular}
\end{center}
\caption{$\eta$ and $\mathrm{pre}_\eta$ for Example~\ref{ex:runningexample} and Fig.~\ref{fig:runningcfg}}
\label{tab:runningexample}
\end{table}

We now define the notion of polynomial ranking-supermartingale.
The intuition is that we encode the RSM-difference condition as a logical formula,
treat zero as the threshold between terminal and non-terminal labels,
and use the invariant $I$ to over-approximate the set of reachable configurations at each label. 
Below for each $\loc\in\locs_\mathrm{c}$, we define $\TERM{\loc}$ to be the propositional polynomial predicate $\bigvee_{(\loc,\phi,\loc')\in\transitions, \loc'\ne\bot}\phi$;
and for $\loc\in\locs\backslash\locs_\mathrm{c}$, we let $\TERM{\loc}:=\TRUE$.

\begin{definition}[Polynomial Ranking-Supermartingale]\label{def:prsm}
A \emph{$d$-degree polyonomial ranking-supermartingale map} (in short, $d$-pRSM) w.r.t $(P,I)$
is a function $\eta: \locs_\bot\times\Rset^{|\pvars|}\rightarrow\Rset$ satisfying that
there exist $\epsilon>0$ and $K\le-\epsilon$ such that for all $\loc\in\locs_\bot$ and all $\vec{x}\in\mathbb{R}^{|\pvars|}$, the conditions (C1-C4) hold:
\begin{itemize}\itemsep1pt \parskip0pt \parsep0pt
\item {\em C1:} the function $\eta(\loc,\cdot):\mathbb{R}^{|\pvars|}\rightarrow\mathbb{R}$
is a polynomial over $\pvars$ of order at most $d$;
\item {\em C2:} if $\loc\ne\bot$ and $\vec{x}\models I(\loc)$, then $\eta(\loc,\vec{x})\ge 0$;
\item {\em C3:} if $\loc=\bot$, then $\eta(\loc, \vec{x})=K$;
\item {\em C4:} if $\loc\ne\bot$ and $\vec{x}\models I(\loc)\wedge\TERM{\loc}$, then $\mathrm{pre}_\eta(\loc,\vec{x})\le\eta(\loc,\vec{x})-\epsilon$ \enskip.
\end{itemize}
\end{definition}
Note that C2 and C3 together separate non-termination and termination by the threshold $0$, and C4 is the {\em RSM difference} condition which is intuitively
related to the $\epsilon$ difference in the RSM definition
(cf. Definition~\ref{def:rsm}).
By generalizing our previous proofs in~\cite{DBLP:conf/popl/ChatterjeeFNH16} (from LRSM to pRSM), 
we establish the soundness of pRSMs w.r.t both almost-sure and finite termination (proof in Appendix~\ref{sect:prsmproof}).

\begin{theorem}
\label{thm:supermartingale-correctness}
If there exists a $d$-pRSM $\eta$ w.r.t $(P,I)$ with constants $\epsilon,K$ (cf. Definition~\ref{def:prsm}), then $P$ is a.s. terminating and $\Eval(P)\le \UB(P):=\frac{\eta(\loc_0,\initval)-K}{\epsilon}$.
\end{theorem}

\begin{example}
\label{ex:prsm}
Consider the running example (cf. Example~\ref{ex:runningexample} and Example~\ref{ex:runningcfg}) and the function $\eta$ given in Example~\ref{ex:preexpectation}.
Assuming that the initial valuation satisfies $1\le x\wedge x\le 10$, 
we assign the trivial invariant $I$ such that $I(1)=0\le x\wedge x\le 11$,
$I(j)=1\le x\wedge x\le 10$ for $2\le j\le 6$ and $I(7)=x<1\vee x>10$. It is straightforward to verify that $\eta$ is a $2$-pRSM 
with $\epsilon=0.2$ and $K=-0.2$ (cf. Definition~\ref{def:prsm} for $\epsilon, K$).  
Hence by Theorem~\ref{thm:supermartingale-correctness}, the program in Example~\ref{ex:runningexample} terminates almost-surely under any scheduler 
and its expected termination time is at most $5\cdot(x_0-1)\cdot(10-x_0)+51$, given the initial value $x_0$. 
\end{example}

\begin{remark}\label{rmk:linearnotsuffice}
The running example (cf. Example~\ref{ex:runningexample} and Example~\ref{ex:runningcfg}) does not admit a linear (i.e. $1$-) pRSM since $\expv_R(r)=0$ at label $3$.
This indicates that linear pRSMs may not exist even over simple affine programs like Example~\ref{ex:runningexample}. 
Thus, this motivates the study of pRSMs even for simple affine programs.
\end{remark}

\begin{remark}\label{rmk:prsmextra}
The non-strict inequality symbol `$\ge$' in C2 can be replaced by its strict counterpart `$>$' since $\eta+c$ ($c>0$) remains to be a pRSM if $\eta$ is a pRSM and $K$ (in C3) is sufficiently small.
(By definition, $\mathrm{pre}_{\eta+c}=\mathrm{pre}_\eta+c$.)
And the non-strict inequality symbol `$\le$' in C4 can be replaced by `$<$' since a pRSM $\eta$ and a constant $K$ (for C3) can be scaled by a constant factor (e.g. $1.1$) so that strict inequalities are ensured.
Moreover, one can also assume that $K=-1$ and $\epsilon=1$ in Definition~\ref{def:prsm}.
This is because one can first scale a pRSM with constants $\epsilon, K$ by a positive scalar to ensure that $\epsilon=1$, and then safely set $K=-1$ due to C2.
\end{remark}

Theorem~\ref{thm:supermartingale-correctness} answers the questions of almost-sure and finite termination in a unified fashion.
Generalizing our approach in~\cite{DBLP:conf/popl/ChatterjeeFNH16}, we show that by restricting a pRSM to have \emph{bounded difference}, we also obtain concentration results.

\begin{definition}[Difference-Bounded pRSM]\label{def:dbprsm} A $d$-pRSM $\eta$ is \emph{difference-bounded} w.r.t a non-empty interval $[a,b]\subseteq\mathbb{R}$ if the following conditions hold:
\begin{compactitem}
\item for all $\loc\in\locs_\mathrm{d}\cup\locs_\mathrm{p}$ and $(\loc,\alpha,\loc')\in\transitions$, and for all $\vec{x}\in \SAT{I(\loc)}$,
it holds that $a\le \eta(\loc',\vec{x})-\eta(\loc,\vec{x})\le b$;
\item for all $\loc\in\locs_\mathrm{c}$ and $(\loc,\phi,\loc')\in\transitions$, and for all $\vec{x}\in \SAT{I(\loc)\wedge\phi}$, it holds that $a\le \eta(\loc',\vec{x})-\eta(\loc,\vec{x})\le b$;
\item for all $\loc\in\locs_\mathrm{a}$ and $(\loc,f,\loc')\in\transitions$, for all $\vec{x}\in \SAT{I(\loc)}$ and for all $\vec{r}\in\{\vec{r}'\mid \forall r\in \rvars.\ \vec{r}'[r]\in\supp_r\}$, it holds that $a\le \eta(\loc',f(\vec{x},\vec{r}))-\eta(\loc,\vec{x})\le b$.
\end{compactitem}
\end{definition}
Note that if a $d$-pRSM $\eta$ with constants $\epsilon,K$ (cf. Definition~\ref{def:prsm}) is difference-bounded w.r.t $[a,b]$, then from definition $a\le -\epsilon$; one can further assume that $-\epsilon\le b$ since otherwise one can reset $\epsilon:=-b$. 
By definition, the stochastic process $X_n:=\eta(\theta_n, \overline{\vec{x}}_n)$ defined through
a difference-bounded pRSM w.r.t $[a,b]$ satisfies that $a\le X_{n+1}-X_n\le b$;
then using Hoeffding's Inequality~\cite{Hoeffding1963inequality,DBLP:conf/popl/ChatterjeeFNH16},
we establish a concentration bound.

\begin{theorem}
\label{thm:concentration}
Let $\eta$ be a difference-bounded $d$-pRSM w.r.t $[a,b]$ with constants $\epsilon$ and $K$. 
For all $n\in\mathbb{N}$, if $\epsilon(n-1)>\eta(\loc_0,\initval)$, then
$\probm(T_P > n)\le e^{-\frac{2(\epsilon(n-1)-\eta(\loc_0,\initval))^2}{(n-1)(b-a)^2}}$\enskip.
\end{theorem}
From Theorem~\ref{thm:concentration}, a difference-bounded $d$-pRSM $\eta$ implies a concentration bound of $\frac{\eta(\loc_0,\initval)}{\epsilon}+2$
(detailed proof of the theorem is in Appendix~\ref{app:concentration}). 

\begin{example}
Consider again our running example in Example~\ref{ex:runningexample} with invariant given in Example~\ref{ex:prsm}. 
Let $\eta$ be the function illustrated in Table~\ref{tab:runningexample}. 
One can verify that the interval $[-10.2 , 8.6]$ satisfies the conditions specified in Definition~\ref{def:dbprsm} for $\eta$, as the following hold:

\begin{compactitem}
\item for all $x\in [1,10]$, $\eta(2,x)-\eta(1,x)=-0.2$; 
\item for all $x\in [0,1)\cup(10,11]$, $-10.2\le\eta(7,x)-\eta(1,x)\le -0.2$;
\item for all $x\in [1,10]$ and $i\in\{3,4\}$, $\eta(i,x)-\eta(2,x)=-0.2$;
\item for all $x\in [1,10]$ and $i\in\{5,6\}$, $-9.4\le\eta(i,x)-\eta(4,x)\le 8.6$;
\item for all $x\in [1,10]$, $\eta(1,x-1)-\eta(5,x)=-0.2$;
\item for all $x\in [1,10]$, $\eta(1,x+1)-\eta(6,x)=-0.2$;
\item for all $x\in [1,10]$ and $r\in\{-1,1\}$, 
\[
-9.6\le \eta(1,x+r)-\eta(3,x)\left(=-2rx-r^2+11r+0.4\right)\le 8.4\enskip.
\]
\end{compactitem}

Then by Theorem~\ref{thm:concentration}, assuming that the program have initial value $x_0=5$, one can deduce that 
\[
\probm\left(T_P>50000\right)\le e^{-\frac{2\cdot (0.2\cdot 49999-30)^2}{49999\cdot 18.8^2}}\approx 1.3016\cdot 10^{-5}\enskip. 
\]
\end{example}

We end this section with a result stating that whether a (difference-bounded) $d$-pRSM exists can be decided 
(proof in Appendix~\ref{app:thmprsmdecision}).
However, the complexity obtained for Theorem~\ref{thm:prsmdecision} is high since it involves quantifier elimination. 
In the next section, we will present efficient algorithms for synthesizing pRSMs.

\begin{theorem}\label{thm:prsmdecision}
For any fixed natural number $d\in\Nset$, the problem whether a (difference-bounded) $d$-pRSM w.r.t an input pair $(P,I)$ exists is decidable.
\end{theorem}

\section{The Synthesis Algorithm}\label{sect:synthesisalgorithm}
In this section, we present an efficient algorithmic approach for solving 
almost-sure/finite termination and concentration questions through synthesis of 
pRSMs. 
Instead of quantifier elimination (of Theorem~\ref{thm:prsmdecision}) we use 
Positivstellensatz (German for ``positive-locus-theorem'' which is related 
to polynomials over semialgebraic sets),
and the approach of this section is sound but not complete (in contrast to 
the computationally expensive but complete method of Theorem~\ref{thm:prsmdecision}).
Note that by Theorem~\ref{thm:supermartingale-correctness}, the existence of a pRSM 
implies both almost-sure and finite termination of a probabilistic program. 

\paragraph{The General Framework.} To synthesize a pRSM, the algorithm first sets up a polynomial template with unknown coefficients.
Next, the algorithm finds values for the unknown coefficients, $\epsilon,K$ (cf. Definition~\ref{def:prsm}) and $[a,b]$ (cf. Definition~\ref{def:dbprsm}) so that C2-C4 in Definition~\ref{def:prsm} and concentration conditions in Definition~\ref{def:dbprsm} are satisfied. 
Note that from Definition~\ref{def:preexpectation}, each $\mathrm{pre}_\eta(\loc,\cdot)$ is a polynomial over $\pvars$ whose coefficients are \emph{linear combinations} of unknown coefficients from the polynomial template. 
Instead of using quantifier elimination (cf. e.g.~\cite{DBLP:journals/fcsc/YangZZX10} or Theorem~\ref{thm:prsmdecision}), 
we utilize Positivstellensatz's~\cite{ScheidererSurvey}. 
We observe that each universally-quantified formula described in C2, C4 and Definition~\ref{def:dbprsm} can be decomposed (through disjunctive normal form of propositional polynomial predicate or transformation of $\max$ in Definition~\ref{def:preexpectation} into two conjunctive clauses) into a conjunction of formulae of the following pattern ($\dag$)
\[
\forall \vec{x}\in\mathbb{R}^{|\pvars|}. \left[\left(g_1(\vec{x})\ge 0\wedge\dots\wedge g_m(\vec{x})\ge 0\right)\rightarrow g(\vec{x})> 0\right]\qquad(\dag)
\]
where each $g_i$ is a polynomial with constant coefficients and $g$ is one 
with unknown coefficients from the polynomial template. 
In the pattern, we over-approximate any possible `$g_j(\vec{x})>0$'  by `$g_j(\vec{x})\ge 0$'. By Remark~\ref{rmk:prsmextra}, the difference between `$g(\vec{x})> 0$' and 
`$g(\vec{x})\ge 0$' does not matter. 

\begin{example}\label{ex:pattern}
Consider again the program in Example~\ref{ex:runningexample} with CFG in 
Example~\ref{ex:runningcfg}. 
Consider the invariant specified in Example~\ref{ex:prsm}. 
The instances of the pattern for termination of this program are listed as follows, where each instance is represented by a pair $(\Gamma,g)$ where $\Gamma$ and $g$ 
corresponds to $\{g_1,\dots,g_m\}$ and resp. $g$ described in ($\dag$). 
\begin{compactitem}
\item (C4, label $1$) $(\{x-1,10-x,x,11-x\}, \eta(1,x)-\eta(2,x)-\epsilon)$;
\item (C4, label $2$) $(\{x-1,10-x\}, \eta(2,x)-\eta(3,x)-\epsilon)$ and $(\{x-1,10-x\}, \eta(2,x)-\eta(4,x)-\epsilon)$;
\item (C4, label $3$) $(\{x-1,10-x\}, \eta(3,x)-\expv_R((y,r)\mapsto\eta(1,y+r), x)-\epsilon)$;
\item (C4, label $4$) $(\{x-1,10-x\}, \eta(4,x)-0.51\eta(5,x)-0.49\eta(6,x)-\epsilon)$;
\item (C4, label $5$) $(\{x-1,10-x\}, \eta(5,x)-\eta(1, x-1)-\epsilon)$;
\item (C4, label $6$) $(\{x-1,10-x\}, \eta(6,x)-\eta(1, x+1)-\epsilon)$;
\item (C2) $(\{x,11-x\}, \eta(1,x))$ and $(\{x-1,10-x\}, \eta(j,x))$ for $2\le j\le 6$.
\end{compactitem}
\end{example}
In the next part, we show that such pattern can be solved by Positivstellensatz's. 

\subsection{Positivstellensatz's}\label{sect:positivstellensatz}

We fix a linearly-ordered finite set $X$ of variables and a finite set $\Gamma=\{g_1,\dots,g_m\}\subseteq\POLYNS{X}$ of polynomials. 
Let $\SAT{\Gamma}$ be the set of all vectors $\vec{x}\in\Rset^{|X|}$ satisfying the propositional polynomial predicate $\bigwedge_{i=1}^m g_i\ge 0$. 
We first define pre-orderings and sums of squares as follows.

\begin{definition}[Sums of Squares]
Define $\Theta$ to be the set of \emph{sums-of-squares}, i.e, 
\[
\Theta:=\left\{\sum_{i=1}^k h^2_i \mid k\in\Nset\mbox{ and }h_{1},\dots,h_k\in\POLYNS{X}\right\}~.
\]
\end{definition}

\begin{definition}[Preordering] 
The \emph{preordering} generated by $\Gamma$ is defined by:
\[
\PO{\Gamma}:=\left\{\sum_{w\in\{0,1\}^m} h_w\cdot \prod_{i=1}^{m} g_i^{w_i}\mid  \forall w.\ h_w\in\Theta\right\}~. 
\]
\end{definition}

\begin{remark}\label{rmk:sumofsquares}
It is well-known that 
a real-coefficient polynomial $g$ of degree $2d$ is a sum of squares iff there exists a $k$-dimensional positive semi-definite real square matrix $Q$ such that $g=\vec{y}^\mathrm{T} Q\vec{y}$, where $k$ is the number of monomials of degree no greater than $d$ and $\vec{y}$ is the column vector of all such monomials (cf.~\cite[Corollary 7.2.9]{MatrixAnalysis}). 
This implies that the problem whether a given polynomial (with real coefficients) is a sum of squares can be solved by semi-definite programming~\cite{SemidefiniteProgramming}. 
\end{remark}

Now we present the first Positivstellensatz, called Schm\"{u}dgen's Positivstellensatz. 

\begin{theorem}[Schm\"{u}dgen's Positivstellensatz~\cite{SchmudgenPositivstellensatz}]
\label{thm:schmuedgen}
Let $g\in\POLYNS{X}$. 
If the set $\SAT{\Gamma}$ is compact and $g(\vec{x})>0$ for all $\vec{x}\in\SAT{\Gamma}$, then $g\in\PO{\Gamma}$. 
\end{theorem}

From Schm\"{u}dgen's Positivstellensatz, any polynomial $g$ which is positive on $\SAT{\Gamma}$ can be represented by 
\[
(\ddag)\qquad g=\sum_{w\in\{0,1\}^m}h_w\cdot g_w~,
\]
where $g_w:=\prod_{i=1}^m g_i^{w_i}$ and $h_w\in\Theta$ for each $w\in\{0,1\}^m$. 
To apply Schm\"{u}dgen's Positivstellensatz, the degrees of those $h_w$'s are restricted to be no greater than a fixed natural number.
Then from Remark~\ref{rmk:sumofsquares} and by equating the coefficients of the same monomials between the two polynomials,
Eq.~($\ddag$) results in a system of linear equalities that involves variables for synthesis of a pRSM and variables (grouped as $2^m$ square matrices) under semi-definite constraints. 

\begin{example}\label{ex:schmuegen}
Assume that $X=\{x\}$ and $\Gamma=\{1-x,1+x\}$. Choose the maximal degree for sums of squares to be $2$. 
Then from Remark~\ref{rmk:sumofsquares}, the form of Eq.~($\ddag$) can be written as: 
\[
g=\sum_{i=1}^4 \left[\begin{pmatrix} 1 & x \end{pmatrix}\cdot \begin{pmatrix} a_{i,1,1} & a_{i,1,2} \\ a_{i,2,1} & a_{i,2,2} \end{pmatrix}\cdot\begin{pmatrix} 1 \\ x\end{pmatrix}\right]\cdot u_i 
\]
where $u_1=1$, $u_2=1-x$, $u_3=1+x$, $u_4=1-x^2$ and each matrix $(a_{i,j,k})_{2\times 2}$ ($1\le i\le 4$) is a matrix of variables subject to be positive semi-definite. 
\end{example}

Theorem~\ref{thm:schmuedgen} can be further refined by a weaker version of Putinar's Positivstellensatz. 

\begin{theorem}[Putinar's Positivstellensatz~\cite{PutinarPositivstellensatz}]
\label{thm:putinar}
Let $g\in\POLYNS{X}$. 
If (i) there exists some $g_i\in\Gamma$ such that the set $\{\vec{x}\in\Rset^{|X|}\mid g_i(\vec{x})\ge 0\}$ is compact and (ii) $g(\vec{x})>0$ for all $\vec{x}\in \SAT{\Gamma}$, then 
\[
(\S)\qquad g=h_0+\sum_{i=1}^m h_i\cdot g_i 
\]
for some sums of squares $h_0,\dots,h_m\in\Theta$. 
\end{theorem}

Likewise, Eq.~($\S$) results in a system of linear equalities that involves variables for synthesis of a pRSM and matrices of variables under semi-definite constraints,
provided that an upper bound on the degrees of sums of squares is enforced. 

\begin{example}\label{ex:putinar}
Assume that $X=\{x\}$ and $\Gamma=\{1-x^2, 0.5-x\}$. Choose the maximal degree for sums of squares to be $2$.  
Then the form of Eq.~(\S) can be written as:
\[
g=\sum_{i=1}^3 \left[\begin{pmatrix} 1 & x \end{pmatrix}\cdot \begin{pmatrix} a_{i,1,1} & a_{i,1,2} \\ a_{i,2,1} & a_{i,2,2} \end{pmatrix}\cdot\begin{pmatrix} 1 \\ x\end{pmatrix}\right]\cdot u_i 
\]
where $u_1=1$, $u_2=1-x^2$, $u_3=0.5-x$ and each matrix $(a_{i,j,k})_{2\times 2}$ ($1\le i\le 4$) is a matrix of variables subject to be positive semi-definite.
\end{example}

In the following, we introduce a Positivstellensatz entitled Handelman's Theorem when $\Gamma$ consists of only linear (degree one) polynomials. 
For Handelman's Theorem, we assume that $\Gamma$ consists of only linear (degree $1$) polynomials
and $\SAT{\Gamma}$ is non-empty. (Note that whether a system of linear inequalities has a solution is decidable in PTIME~\cite{DBLP:books/daglib/0090562}.)

\begin{definition}[Monoid]
The \emph{monoid} of $\Gamma$ is defined by:
\[
\mbox{\sl Monoid}(\Gamma):=\left\{\prod_{i=1}^k h_i \mid k\in\Nset_0\mbox{ and }h_1,\dots,h_k\in\Gamma\right\}~~.
\]
\end{definition}

\begin{theorem}[Handelman's Theorem~\cite{HandelmanTheorem}]
\label{thm:handelman}
Let $g\in\POLYNS{X}$ be a polynomial such that $g(\vec{x})>0$ for all $\vec{x}\in\SAT{\Gamma}$. 
If $\SAT{\Gamma}$ is compact, then 
\[
(\#)\qquad g=\sum_{i=1}^d a_i\cdot u_i 
\]
for some $d\in\Nset$, real numbers $a_1,\dots,a_d\ge 0$ and $u_1,\dots,u_d\in\mbox{\sl Monoid}(\Gamma)$. 
\end{theorem}

To apply Handelman's theorem, we consider a natural number which serves as a bound on the number of multiplicands allowed to form an element in $\mbox{\sl Monoid}(\Gamma)$;
then Eq.~($\#$) results in a system of linear equalities involving $a_1,\dots,a_d$.
Unlike previous Positivstellensatz's, the form of Handelman's theorem allows us to construct a system of linear equalities free from semi-definite constraints. 

\begin{example}\label{ex:handelman}
Assume that $X=\{x\}$ and $\Gamma=\{1-x,1+x\}$. 
Fix the maximal number of multiplicands in an element of $\mbox{\sl Monoid}(\Gamma)$ to be $2$. 
Then the form of Eq.~($\#$) can be rewritten as 
\[
g=\sum_{i=1}^6 a_i\cdot u_i
\]
where $u_1=1$, $u_2=1-x$, $u_3=1+x$, $u_4=1-x^2$, $u_5=1-2x+x^2$, $u_6=1+2x+x^2$ and each $a_i$ ($1\le i\le 6$) is subject to be a non-negative real number. 
\end{example}

\subsection{The Algorithm for pRSM Synthesis}

Based on the Positivstellensatz's introduced in the previous part, we present our algorithm for synthesis of pRSMs. 
Below, we fix an input probabilistic program $P$, an input polynomial invariant $\inv$ and an input initial configuration $(\loc_0,\initval)$ for $P$. 
Let $\mathcal{G}=(\locs,\bot,(\pvars,\rvars),\transitions)$ be the associated CFG of $P$.

\smallskip\noindent{\em Description of the Algorithm \PRSMSynth{}.} 
We present a succinct description of the key ideas. 
The description of the key steps of the algorithm is as follows. 
\begin{compactenum}
\item {\em Template $\eta$ for a pRSM.} The algorithm fix a natural number $d$ as the maximal degree for a pRSM, construct $\mathcal{M}_d$ as the set of all monomials over $X$ of degree no greater than $d$,
and set up a template $d$-pRSM $\eta$ such that $\eta(\loc,\cdot)$ is the polynomial $\sum_{h\in\mathcal{M}_d} a_{h,\loc}\cdot h$ where each $a_{h,\loc}$ is a (distinct) scalar variable (cf. C1).
\item {\em Bound for Sums of Squares and Monoid Multiplicands.} The algorithm fix a natural number $k$ as the maximal degree for a sum of squares (cf. Schm\"{u}dgen's and Putinar's Positivstellensatz) 
or as the maximal number of multiplicands in a monoid element (cf. Handelman's Theorem). 
\item {\em RSM-Difference and Terminating-Negativity.} 
From Remark~\ref{rmk:prsmextra}, the algorithm fixes $\epsilon$ to be $1$ (cf. condition C3) and $K$ to be $-1$ (cf. condition C4). 
\item {\em Computation of pre-expectation $\mathrm{pre}_\eta$ .} With $\epsilon,K$ fixed to be resp. $1,-1$ in the previous step, the algorithm computes $\mathrm{pre}_\eta$ by Definition~\ref{def:preexpectation}, whose all involved coefficients are linear combinations from $a_{h,\loc}$'s. 
\item {\em Pattern Extraction.} The algorithm extracts instances conforming to pattern ($\dag$) from C2, C4 and formulae presented in Definition~\ref{def:dbprsm}, and translates them into systems of linear equalities 
over variables among $a_{h,\loc}$'s, $\epsilon$, $K$, and extra matrices of variables assumed to be positive semi-definite (cf. 
Schm\"{u}dgen's and Putinar's Positivstellensatz) or scalar variables assumed to be non-negative (cf. Handelman's Theorem) through Eq.~($\ddag$), Eq.~($\S$) and Eq.~($\#$). 
\item {\em Solution via Semidefinite or Linear Programming.} The algorithm calls semi-definite programming (for Schm\"{u}dgen's and Putinar's Positivstellensatz) or linear programming (for Handelman's Theorem) in order to check the feasibility or to optimize $\UB(P)$ (cf. Theorem~\ref{thm:supermartingale-correctness} for
upper bound of $\Eval(P)$) over all variables among $a_{h,\loc}$'s and extra matrix/scalar variables from Eq.~($\ddag$), Eq.~($\S$) and Eq.~($\#$). 
Note that the feasibility implies the existence of a (difference-bounded) $d$-pRSM; 
the existence of a $d$-pRSM in turn implies finite termination, and the existence of a difference-bounded $d$-pRSM in turn implies a concentration bound through Theorem~\ref{thm:concentration}.
\end{compactenum}

The soundness of our algorithm is as follows, whose proof is in Appendix~\ref{app:thmsynthesissoundness}. 

\begin{theorem}[Soundness]\label{thm:synthesissoundness}
Any function $\eta$ synthesized through the algorithm \PRSMSynth{} is a valid pRSM.   
\end{theorem}

\begin{remark}[Efficiency]\label{rmk:polycase}
It is well-known that for semi-definite programs with a positive real number $R$ to bound the Frobenius norm of any 
feasible solution, an approximate solution upto  precision $\epsilon$ can be computed in polynomial time in the size of 
the semi-definite program (with rational numbers encoded in binary), $\log R$ and $\log\epsilon^{-1}$~\cite{SemidefiniteProgramming}. 
Thus, our sound approach presents an efficient method for analysis of many probabilistic programs. 
Moreover, when each propositional polynomial predicate in the probabilistic program involves only linear polynomials, 
then the sound form of Handelman's theorem can be applied, resulting in feasibility checking of systems of linear inequalities 
rather than semi-definite constraints. 
By polynomial-time algorithms for solving systems of linear inequalities~\cite{DBLP:books/daglib/0090562}, our approach is 
polynomial time (and thus efficient) over such programs.
\end{remark}

\begin{remark}[Semi-Completeness]\label{rmk:semicompleteness}
Consider probabilistic programs of the following form:
$\textbf{while}~\phi~\textbf{do}~\textbf{if $\star$ then } P_1~\textbf{else } P_2~\textbf{od}$~,
where $P_1,P_2$ are single assignments, $\SAT{\phi}$ is compact, and invariants which assign to each label a propositional polynomial predicate is in DNF form that involves no strict inequality (i.e. no `$<$' or `$>$'). Upon such inputs, our approach is \emph{semi-complete} in the sense that by raising the upper bounds for the degree of a sum of squares and the number of multiplicands in a monoid element,  
the algorithm \PRSMSynth{} will eventually find a pRSM if it exists. 
This is because Theorem~\ref{thm:schmuedgen} to~\ref{thm:handelman} are ``semi-complete'' when $\SAT{\Gamma}$ is compact, as the terminal label can be separately handled by $\TERM{\cdot}$ so that only compact $\Gamma$'s for Positivstellensatz's may be formed, and the difference between strict and non-strict inequalities does not matter (cf. Remark~\ref{rmk:prsmextra}).  
\end{remark}

\begin{remark}[Comparision with our previous result~\cite{DBLP:conf/popl/ChatterjeeFNH16}]
Our approach using Handelman's theorem is a strict generalization of the LRSM (linear ranking supermartingale) 
approach of~\cite{DBLP:conf/popl/ChatterjeeFNH16} that uses Farkas' lemma.
For example, our approach using Handelman's Theorem applied to affine programs can 
handle Example~\ref{ex:runningexample}, where no LRSM exists (Remark~\ref{rmk:linearnotsuffice}).
\end{remark}

\begin{remark}[New techniques for nonprobabilistic programs]
To the best of our knowledge, Schm\"{u}dgen's Positivstellensatz
and Handelman's Theorem have not been used for nonprobabilistic programs,
and thus our approach presents new analysis methods even for 
nonprobabilistic programs (though our approach is for the more general class
of nondeterministic probabilistic programs).
\end{remark}

\begin{remark}[Key Insights]
The key insights of this paper are (i) the need for pRSMs (cf. Remark~\ref{rmk:linearnotsuffice}), 
(ii) the adaptation of conditional expectation with pRSMs, (iii) the connection between synthesis of pRSMs and Positivstellensatz's and (iv) the adoption of semidefinite and linear programming to synthesize pRSMs. 
\end{remark}

\section{Experimental Results}\label{sect:experimental}
In this section, we present experimental results for our algorithm
through the semi-definite programming tool SOSTOOLS~\cite{sostools} 
(that uses SeDuMi~\cite{sedumi}) and 
the linear programming tool CPLEX~\cite{cplex}.
Due to space constraints, the detailed description of the input probabilistic 
programs are in Appendix~\ref{app:experiments}. 

\smallskip\noindent{\em Experimental examples and setup.}
We consider six classical examples of probabilistic programs that exhibit distinct types non-linear 
behaviours.
Our examples are, namely, \emph{Logistic Map} 
adopted in~\cite{DBLP:conf/vmcai/Cousot05} (Example~\ref{ex:logisticmap} in Appendix~\ref{app:experiments}) 
which was previously handled by Lagrangian relaxation and semi-definite programming 
whereas our approach is polynomial time using linear programming,
\emph{Decay} that models a sequence of points converging stochastically to the origin 
(Example~\ref{ex:decay} in Appendix~\ref{app:experiments}), 
\emph{Random Walk} that models a random walk within a bounded region defined through non-linear curves (Example~\ref{ex:randomwalk} 
in Appendix~\ref{app:experiments}), 
\emph{Gambler's Ruin} which is our running example (Example~\ref{ex:runningexample}), 
\emph{Gambler's Ruin Variant} (Example~\ref{ex:gamblersruin2} in Appendix~\ref{app:experiments}) 
which is a variant of (Example~\ref{ex:gamblersruin}), and 
\emph{Nested Loop} (Example~\ref{ex:nestedloop} in Appendix~\ref{app:experiments}) which is a nested loop 
with stochastic increments. 
Except for \emph{Gambler's Ruin Variant} and \emph{Nested Loop}, our approach is semi-complete for all other 
examples (cf. Remark~\ref{rmk:semicompleteness}). 
In all the examples the invariants are straightforward and was manually integrated with the input. 
Since SOSTOOLS only produces numerical results, we modify C2 to ``$\eta(\loc,\vec{x})\ge 0$'' to ``$\eta(\loc,\vec{x})\ge 1$'' 
for Putinar's or Schm\"{u}dgen's Positivstellensatz and check whether the maximal numerical error of all equalities added to SOSTOOLS 
is sufficiently small over a bounded region. 
In our examples, the bounded region is $\{(x,y)\mid x^2+y^2\le 2\}$ (cf.~Example~\ref{ex:randomwalk} and Example~\ref{ex:decay}) 
and the maximal numerical error should not exceed $1$. 
Note that $1$ is also our fixed $\epsilon$ in C4, and by Remark~\ref{rmk:prsmextra}, the modification on C2 is not 
restrictive.
Instead, one may also pursue Sylvester's Criterion (cf.~\cite[Theorem 7.2.5]{MatrixAnalysis}) to check membership of sums of squares through checking whether a square matrix is positive semi-definite or not.
More elegant approaches for numerical problems is a subject of future work.

\smallskip\noindent{\em Experimental results.}
In Table~\ref{tab:experiments}, we present the experimental results, where 
`Method' means that whether we use either Handelman's Theorem, Putinar's Positivstellensatz or Schm\"{u}dgen's Positivstellensatz to synthesize pRSMs, 
`SOSTOOLS/CPLEX' means the running time for CPLEX/SOSTOOLS in seconds, 'error' is the maximal numerical error of equality constraints added into SOSTOOLS 
(when instantiated with the solutions), and $\eta(\loc_0,\cdot)$ is the polynomial for the initial label in the synthesized pRSM. 
The synthesized pRSMs (in the last column) refer to the variables of the program.
All numbers except errors are rounded to $10^{-4}$.
For all the examples, our translation to the optimization problems are linear.
We report the running times of the optimization tools and synthesized pRSMs.
The experimental results were obtained on Intel Core i7-2600 machine with 3.4 GHz 
with 16GB RAM.

\begin{table}
\begin{center}
{\footnotesize
\begin{tabular}{|c|c|c|c|c|}
\hline
  Example &  Method &  SOSTOOLS & error & $\eta(\loc_0,\cdot)$   \\
\hline
Decay & Putinar & $0.1248$s  & $\le 10^{-9}$ & $ 5282.3435x^2 + 5282.3435y^2 + 1$\\
\hline
Random Walk  & Schm\"{u}dgen & $0.7176$s & $\le 10^{-7}$ & $-300x^2 - 300y^2 + 601$\\
\hline
\hline
  Example &  Method &  CPLEX & - &  $\eta(\loc_0,\cdot)$   \\
\hline
Gambler's Ruin & Handelman & $\le 10^{-2}$s  & - & $33x-3x^2$ \\
\hline
Gambler's Ruin V. & Handelman & $\le 10^{-2}$s  & - & $-21+100x-70y-100x^2+100xy$ \\
\hline
Logistic Map  & Handelman & $\le 10^{-2}$s & - & $1000500.7496x$ \\
\hline
Nested Loop & Handelman & $\le 2\cdot10^{-2}$s & - &  $48 + 160n  + (m-x)(800n+240) $ \\
\hline
\hline
\end{tabular}
}
\end{center}
\caption{Experimental Results}
\label{tab:experiments}
\end{table}

For all the examples we consider except Logistic Map, their almost-sure termination cannot be answered by previous approaches. 
For the Logistic-Map example, our reduction is to linear programming whereas existing approaches~~\cite{DBLP:conf/vmcai/Cousot05,DBLP:journals/jossac/ShenWYZ13} reduce to semidefinite programming.

\section{Conclusion and Future Work}\label{sect:conclusion}
In this paper, we extended linear ranking supermartingale (LRSM) for probabilistic 
programs proposed  in~\cite{SriramCAV,DBLP:conf/popl/ChatterjeeFNH16} to polynomial ranking supermartingales (pRSM) 
for nondeterministic probabilistic programs. 
We developed the notion of (difference bounded) pRSM and proved that it is sound for almost-sure and finite termination,
as well as for concentration bound (Theorem~\ref{thm:supermartingale-correctness} and Theorem~\ref{thm:concentration}).
Then we developed an efficient (sound but not complete) algorithm for synthesizing pRSMs 
through Positivstellensatz's (cf.~Sect.~\ref{sect:positivstellensatz}), 
proved its soundness (Theorem~\ref{thm:synthesissoundness}) and argued its 
semi-completeness (Remark~\ref{rmk:semicompleteness}) over an important class 
of programs. 
Finally, our experiments demonstrate the effectiveness of our synthesis approach over various classical probabilistic of programs, 
where LRSMs do not exist (cf. Example~\ref{ex:runningexample} and Remark~\ref{rmk:linearnotsuffice}).   
Directions of future work are to explore (a)~more elegant methods for numerical problems related
to semi-definite programming, and (b)~other forms of RSMs for more general class of 
probabilistic programs. 

\subsubsection*{Acknowledgement}
We thank anonymous referees for valuable comments. 
We also thank Hui Kong for his help on SOSTOOLS. 
The research was partly supported by Austrian Science Fund (FWF) NFN Grant No S11407-N23 (RiSE/SHiNE), ERC Start grant (279307: Graph Games), 
ERC Advanced Grant ( 267989: QUAREM), and Natural Science Foundation of China (NSFC) under
Grant No. 61532019.

\clearpage

\appendix

\section{Propositional Polynomial Predicates}\label{sect:predicate}

Formally, the set of propositional polynomial predicates over $X$ is defined as
the smallest set satisfying the following conditions:
\begin{compactenum}
\item each polynomial constraint over $X$ is a propositional polynomial predicate;
\item both $\TRUE$ and $\FALSE$ are propositional polynomial predicates;
\item if $\phi$ is a propositional polynomial predicate, then so is $\neg\phi$;
\item if $\phi,\psi$ are propositional polynomial predicates, then so are $\phi\wedge\psi$ and $\phi\vee\psi$.
\end{compactenum}
The satisfaction relation $\models$ between real vectors $\vec{x}$ and propositional polynomial predicates $\phi$ is defined by:
\begin{compactitem}
\item $\vec{x}\models\TRUE$ and $\vec{x}\not\models\FALSE$ for all vectors $\vec{x}$;
\item $\vec{x}\models g_1\Join g_2$ iff $g_1(\vec{x}) \Join g_2(\vec{x})$~;
\item $\vec{x}\models \neg\phi$ iff $\vec{x}\not\models\phi$~;
\item $\vec{x}\models \phi\wedge\psi$ iff $\vec{x}\models \phi$ and $\vec{x}\models\psi$~;
\item $\vec{x}\models \phi\vee\psi$ iff $\vec{x}\models \phi$ or $\vec{x}\models\psi$~.
\end{compactitem}

\section{Probabilistic Programs: Detailed Syntax}\label{app:syntax}

Let $\mathcal{X}$ and $\mathcal{R}$ be the disjoint countable collections of \emph{program} and \emph{sampling} variables, respectively.
We assume that each sampling variable $r$ be associated with a one-dimensional cumulative distribution function $\Upsilon_r$ and 
a non-empty bounded interval $\support{r}$ in $\Rset$ such that $\Upsilon_r(\sup\support{r})=1$ and $\Upsilon_r(\inf\support{r})=0$, and 
the sampled values for $r$  fall in $\support{r}$ with probability~1 (this is the rigorous condition of the boundedness of the 
sampling variables). 

\paragraph*{The Syntax.} The syntax of probabilistic programs is given by the grammar in Figure~\ref{fig:syntax}. 
The expressions $\langle \mathit{pvar}\rangle,\langle\mathit{rvar}\rangle$ and $\langle \mathit{pvarlist}\rangle$ range over $\mathcal{X},\mathcal{R}$ and finite sequences of program variables, respectively.
The expressions $\langle \mathit{expr} \rangle, \langle\mathit{rexpr}\rangle$ and $\langle\mathit{rexprlist}\rangle$ may be evaluated to any polynomial with variables in $\mathcal{X}$, any polynomial with variables in $\mathcal{X}\cup\mathcal{R}$ and any finite list of polynomials with variables in $\mathcal{X}\cup\mathcal{R}$, respectively. 
The assignment statement $\langle\mathit{pvarlist}\rangle \,\text{'$:=$'}\, \langle\mathit{rexprlist}\rangle$ specifies simultaneous update of program variables in $\langle\mathit{pvarlist}\rangle$ by polynomial $\langle\mathit{rexprlist}\rangle$ in sequel; we thus assume that each instance of $\langle\mathit{pvarlist}\rangle$ will contain distinct program variables and the length of each instance of $\langle\mathit{pvarlist}\rangle$ will always be equal to the corresponding instance of $\langle\mathit{rexprlist}\rangle$. 
From the assignment statement one observes that sampling variables can only be used in the RHS of an assignment.
Sequential composition, if-branch and while-loop are indicated by semicolon, the keyword \textbf{if} and the keyword \textbf{while}, respectively.  
Moreover, $\langle\mathit{bexpr}\rangle$ may be evaluated to any propositional polynomial predicate. 

\paragraph*{The ``if-then-else'' Statement.} The guard $\langle\mathit{ndbexpr}\rangle$ of each if-then-else statement is either a keyword $\star$ representing demonic resolution of non-determinism, or a keyword 
$\mbox{\textbf{prob(}}p\mbox{\textbf{)}}$ ($p\in (0,1)$ being a number given in decimal representation) representing the probabilistic choice that the if-branch is executed with probability $p$ and the then-branch with probability $1-p$, or a propositional polynomial predicate, in which case the statement represents a standard deterministic conditional branching.

\begin{figure}
\centering
\begin{align*}
&\langle \mathit{stmt}\rangle ::= \,\langle\mathit{pvarlist}\rangle \,\text{'$:=$'}\, \langle\mathit{rexprlist} \rangle \\
&\mid   \text{'\textbf{if}'} \, \langle\mathit{ndbexpr}\rangle\,\text{'\textbf{then}'} \, \langle \mathit{stmt}\rangle \, \text{'\textbf{else}'} \, \langle \mathit{stmt}\rangle \,\text{'\textbf{fi}'}
\\
&\mid  \text{'\textbf{while}'}\, \langle\mathit{bexpr}\rangle \, \text{'\textbf{do}'} \, \langle \mathit{stmt}\rangle \, \text{'\textbf{od}'}
\\ 
&\mid \langle\mathit{stmt}\rangle \, \text{';'} \, \langle \mathit{stmt}\rangle \mid \text{'\textbf{skip}'}
\\
\vspace{\baselineskip}
\\
&\langle\mathit{expr} \rangle ::= \langle \mathit{constant} \rangle \mid \langle\mathit{pvar}\rangle 
\\
& \mid \langle \mathit{expr} \rangle \,\text{'$*$'} \, \langle\mathit{expr}\rangle  
\\
&\mid \langle\mathit{expr} \rangle\, \text{'$+$'} \,\langle\mathit{expr} \rangle \mid \langle\mathit{expr} \rangle\, \text{'$-$'} \,\langle\mathit{expr} \rangle
\\
\vspace{\baselineskip}
\\
&\langle\mathit{rexpr} \rangle ::= \langle \mathit{constant} \rangle \mid \langle\mathit{pvar}\rangle \mid  \langle\mathit{rvar}\rangle
\\
& \mid \langle \mathit{rexpr} \rangle \,\text{'$*$'} \, \langle\mathit{rexpr}\rangle  
\\
& \mid \langle\mathit{rexpr} \rangle\, \text{'$+$'} \,\langle\mathit{rexpr} \rangle \mid \langle\mathit{rexpr} \rangle\, \text{'$-$'} \,\langle\mathit{rexpr} \rangle
\\
\vspace{\baselineskip}
\\
&\langle\mathit{pvarlist}\rangle ::=\langle\mathit{pvar}\rangle\,\text{'$,$'}\,\langle\mathit{pvarlist}\rangle\mid \langle\mathit{pvar}\rangle
\\
&\langle\mathit{rexprlist}\rangle ::=\langle\mathit{rexpr}\rangle\,\text{'$,$'}\,\langle\mathit{rexprlist}\rangle\mid \langle\mathit{rexpr}\rangle
\\
\vspace{\baselineskip}
\\
&\langle\mathit{literal} \rangle ::= \langle\mathit{expr} \rangle\, \text{'$\leq$'} \,\langle\mathit{expr} \rangle \mid \langle\mathit{expr} \rangle\, \text{'$\geq$'} \,\langle\mathit{expr} \rangle
\\
&\langle\mathit{polyexpr} \rangle ::=  \langle\mathit{literal} \rangle\mid \langle\mathit{literal} \rangle\, \text{'\textbf{and}'} \,\langle\mathit{polyexpr} \rangle
\\
&\langle \mathit{bexpr}\rangle ::=  \langle \mathit{polyexpr} \rangle \mid \langle \mathit{polyexpr} \rangle \, \text{'\textbf{or}'} \, \langle\mathit{bexpr}\rangle
\\
&\langle\mathit{ndbexpr} \rangle ::=  \text{'$\star$'}\mid \text{'\textbf{prob($p$)}'} \mid \langle\mathit{bexpr} \rangle
\end{align*}
\caption{Syntax of Probabilistic Programs}
\label{fig:syntax}
\end{figure}

\section{Transformation from Programs to CFGs}\label{app:cfg}

Below we fix a set $\pvars$ of program variables and a set $\rvars$ of sampling variables.
We also fix two linear orders on $\pvars$ and $\rvars$ under which
$\pvars = \{x_1,\dots,x_{|\pvars|}\}$ and $\rvars = \{r_1,\dots,r_{|\pvars|}\}$.

We recall that a \emph{valuation} of program variables is a vector $\vec{x}\in\Rset^{|\pvars|}$ interpreted in the way that the actual value held by a program variable $x_i$ ($1\le i\le |\pvars|$) is $\vec{x}[i]$;
similarly, a \emph{valuation} of sampling variables is a vector $\vec{r}\in\Rset^{|\rvars|}$ such that the sampled value held by $r_i$ is $\vec{r}[i]$.
Every update function $f$ in a CFG can then be viewed as a tuple $(f_1,\dots,f_{|\pvars|})$, where each $f_i$ is of type $\Rset^{|\pvars|}\times\Rset^{|\rvars|}\rightarrow\Rset$.
We use the following succinct notation for special update functions: by $\id$ we denote the function which does not change the program variables at all, i.e. for every $1\leq i \leq |\pvars|$ we have $f_i(\vec{x},\vec{r})= \vec{x}[i]$. For any $k$ functions $g_1,\dots,g_k:\Rset^{|\pvars|}\times\Rset^{|\rvars|}\rightarrow\Rset$ and any sequence $n_1,\dots,n_k$ of
$k$ distinct numbers in $\{1,\dots,|\pvars|\}$, we denote by $\assgn{\{x_{n_j}\}_{1}^k}{\{g_j\}_{1}^k}$ the update function $f$ such that $f_{n_j}(\vec{x},\vec{r})=g_j(\vec{x},\vec{r})$ for $1\le j\le k$ and $f_i(\vec{x},\vec{r})=\vec{x}[i]$ whenever $i\not\in\{n_1,\dots,n_k\}$.

\paragraph*{From Probabilistic Programs to CFGs.}
To every probabilistic program $P$ with programs variables from $\pvars$ and sampling variables from $\rvars$,
we construct a CFG $\mathcal{G}_P$ inductively on the structure of $P$.
The CFG $\mathcal{G}_P$ has $\pvars$ and resp. $\rvars$ as its set of program and resp. sampling variables.
For each program $P$, the CFG $\mathcal{G}_P$ involves two distinguished labels, namely $\loc^{\lin}_{P}$ and $\loc^{\lout}_{P}$, that intuitively represent the label assigned to the first statement to be executed in $P$ and the terminal label of $P$, respectively.
The inductive construction is as follows.
\begin{compactenum}
\item {\em Asssignments and Skips.}
For $P= x_{n_1},\dots,x_{n_k}:= E_1,\dots,E_k$ or resp. $P = \textbf{skip}$, the CFG $\mathcal{G}_P$ consists of a new assignment labels $\loc^{\lin}_{P}$ and a new terminal label $\loc^{\lout}_{P}$,
and a transition $(\loc^{\lin}_{P},\assgn{\{x_{n_j}\}_1^k}{\{E_j\}_1^k},\loc^{\lout}_P)$ or resp. $(\loc^{\lin}_{P},\id,\loc^{\lout}_P)$, where
we treat each $E_j$ as a function through direct evaluation of variables.

\item {\em Sequential Statements.}
For $P = Q_1;Q_2$, we take the disjoint union of the CFGs $\mathcal{G}_{Q_1}$, $\mathcal{G}_{Q_2}$, while redefining $\loc^{\lout}_{Q_1}$ to be $\loc^{\lin}_{Q_2}$ and putting $\loc^{\lin}_{P}=\loc^{\lin}_{Q_1}$ and $\loc^{\lout}_{P}=\loc^{\lout}_{Q_2}$.

\item {\em While Statements.}
For $P = \textbf{while $\phi$ do }Q \textbf{ od}$, we add a new terminal label $\loc^{\lout}_{P}$, change $\loc^{\lout}_{Q}$ to a branching label, add transitions $(\loc^{\lout}_{Q},\phi,\loc^{\lin}_{Q})$ and $(\loc^{\lout}_{Q},\neg\phi,\loc^{\lout}_{P})$, and define $\loc^{\lin}_{P}:=\loc^{\lout}_{Q}$.

\item {\em If Statements.}
For $P = \textbf{if $B$ then }Q_1 \textbf{ else } Q_2 \textbf{ fi}$,
we consider different cases on $B$: if $B$ is some $\textbf{prob}(p)$, then we add a new probabilistic label $\loc^{\lin}_{P}$ together with two transitions
$(\loc^{\lin}_{P},p,\loc^{\lin}_{Q_1})$ and $(\loc^{\lin}_{P},1-p,\loc^{\lin}_{Q_2})$;
if $B$ is some propositional polynomial predicate $\phi$ then we add a new branching label
$\loc^{\lin}_{P}$ together with transitions
$(\loc^{\lin}_{P},\phi,\loc^{\lin}_{Q_1})$ and $(\loc^{\lin}_{P},\neg\phi,\loc^{\lin}_{Q_2})$;
otherwise, $B=`\star'$ and we add a new demonic label $\loc^{\lin}_{P}$ together with  transitions $(\loc^{\lin}_{P},\star,\loc^{\lin}_{Q_1})$ and $(\loc^{\lin}_{P},\star,\loc^{\lin}_{Q_2})$.
In any of the cases above, we also add a new terminal label $\loc^{\lout}_{P}$ and identify both $\loc^{\lout}_{Q_1}$ and $\loc^{\lout}_{Q_2}$ with $\loc^{\lout}_{P}$.
\end{compactenum}

\section{The Semantics: Detailed Description}\label{app:semanticsdetails}

The behaviour of a probabilistic program $P$ accompanied with its CFG $\mathcal{G}=(\locs,\bot,(\pvars,\rvars),\transitions)$ under a scheduler $\sigma$ is described as follows. 
The program starts in the initial configuration $(\loc_0,\initval)$. 
Then in each \emph{step} $i$ ($i\in\mathbb{N}_0$), given the current configuration $(\loc_{i},\vec{x}_{i})$, the next configuration $(\loc_{i+1},\vec{x}_{i+1})$ is determined by the following procedure:
\begin{compactenum}
\item
a valuation $\vec{r}_i$ of the sampling variables is sampled according to the joint distribution of the cumulative distributions $\{\Upsilon_r\}_{r\in R}$ and independent of all previously-traversed configurations (including $(\loc_i,\vec{x}_i)$), all previous samplings on $R$ and previous executions of probabilistic branches;
\item
if $\loc_i\in\locs_\mathrm{d}$ and $c_0,\dots,c_i$ is the finite path traversed so far
(i.e., $c_0=(\loc_0,\initval)$ and $c_i=(\loc_i,\vec{x}_i)$) with $\sigma(c_0,\dots,c_i)=(\loc_i, \star, \loc')$, then 
$(\loc_{i+1},\vec{x}_{i+1})$ is set to be $(\loc',\vec{x}_i)$;
\item 
if $\loc_i\in\locs_\mathrm{p}$ and $(\loc_i, p, \loc_1),(\loc_i, 1-p,\loc_2)$ are namely the two transitions in $\transitions$ with source label $\loc_i$, then with a Bernoulli experiment independent of all previous samplings, probabilistic branches and traversed configurations, $(\loc_{i+1},\vec{x}_{i+1})$ is set to be (i) $(\loc_1,\vec{x}_i)$ with probability $p$ and (ii) $(\loc_2,\vec{x}_i)$ with probability $1-p$; 
\item
if $\loc_i\in\locs_\mathrm{c}$ and $(\loc_i, \phi, \loc_1),(\loc_i, \neg\phi,\loc_2)$ are namely the two transitions in $\transitions$ with source label $\loc_i$, then 
$(\loc_{i+1},\vec{x}_{i+1})$ is set to be (i) $(\loc_1,\vec{x}_i)$ when $\vec{x}_i\models\phi$ and (ii) $(\loc_2,\vec{x}_i)$ when $\vec{x}_i\models\neg\phi$;
\item 
if $\loc_i\in\locs_\mathrm{a}$ and $(\loc_i, f, \loc')$ is the only transition in $\transitions$ with source location $\loc_i$, then $(\loc_{i+1},\vec{x}_{i+1})$ is set to be 
$(\loc',f(\vec{x}_{i},\vec{r}_i))$;
\item 
if $\loc_i=\bot$ then $(\loc_{i+1}, \vec{x}_{i+1})$ is set to be $(\loc_i,\vec{x}_{i})$.
\end{compactenum}

\section{Proof of Lemma ~\ref{lemma:condexp} and Theorem~\ref{thm:supermartingale-correctness}}\label{sect:prsmproof}

\textbf{Lemma~\ref{lemma:condexp}.}
Let $\eta: \locs_\bot\times\mathbb{R}^{|\pvars|}\rightarrow\mathbb{R}$ be a function such that
each $\eta(\loc,\cdot)$ (for all $\loc\in\locs_\bot$) is a polynomial function over $\pvars$, and $\sigma$ be any scheduler.
Let the stochastic process $\{X_n\}_{n\in\mathbb{N}_0}$ be defined by:
$X_{n}:=\eta(\theta_{n},\overline{\vec{x}}_{n})$.
Then for all $n\in\mathbb{N}_0$, we have $\expv^{\sigma}(X_{n+1}\mid\mathcal{H}_n)\le\mathrm{pre}_\eta(\theta_{n},\overline{\vec{x}}_{n})$.
\begin{proof}
For all $n\in\mathbb{N}_0$, from the syntax and semantics of probabilistic program we have
\[
X_{n+1}=\mathbf{1}_{\theta_{n}=\bot}\cdot X_n+Y_\mathrm{p}+Y_\mathrm{d}+Y_\mathrm{c}+Y_\mathrm{a}
\]
where the terms are described below.
\[
Y_\mathrm{p}:=\sum_{\loc\in\locs_\mathrm{p}}\left[\mathbf{1}_{\theta_n=\loc}\cdot \sum_{i\in\{0,1\}} \left(\mathbf{1}_{Z_\loc=i}\cdot\eta(\loc_{Z_\loc=i},\overline{\vec{x}}_{n})\right)\right]
\]
where each random variable $Z_\loc$ is the Bernoulli random variable for the decision of the probabilistic branch and $\loc_{Z_\loc=0},\loc_{Z_\loc=1}$ are the corresponding target labels from $\loc$ in $\transitions$.
(Note that all $Z_\loc$'s are independent of $\mathcal{H}_n$.)
In other words, $Y_\mathrm{p}$ describes the semantics of statements with probabilistic labels.
\[
Y_\mathrm{a}:=\sum_{\loc\in\locs_\mathrm{a}}\left[\mathbf{1}_{\theta_n=\loc}\cdot \eta(\loc',f_\loc(\overline{\vec{x}}_{n},\overline{\vec{r}}_n))\right]
\]
where $(\loc,f_\loc,\loc')$ is the only transition in $\transitions$ with source label $\loc$, describing the semantics of statements with assignment labels.
\[
Y_\mathrm{c}:=\sum_{\loc\in\locs_\mathrm{c}}\sum_{(\loc,\phi,\loc'')\in\transitions} \left[\mathbf{1}_{\theta_n=\loc\wedge \overline{\vec{x}}_{n}\models\phi}\cdot \eta(\loc'', \overline{\vec{x}}_{n})\right]
\]
which describes the semantics of statements with branching labels.
\[
Y_\mathrm{d}:=\sum_{\loc\in\locs_\mathrm{d}}\mathbf{1}_{\theta_n=\loc}\cdot\eta\left(
\mathrm{tgt}\left[\sigma\left(\left\{(\theta_k,\overline{\vec{x}}_{k})\right\}_{0\le k\le n}\right)\right],\overline{\vec{x}}_{n}\right)
\]
where
$\mathrm{tgt}\left[\sigma\left(\left\{(\theta_k,\overline{\vec{x}}_{k})\right\}_{0\le k\le n}\right)\right]$
is the target label of the transition $\sigma\left(\left\{(\theta_k,\overline{\vec{x}}_{k})\right\}_{0\le k\le n}\right)$,
describing the semantics of demonic labels.
Then from properties of conditional expectation~\cite[Page 88]{probabilitycambridge}, one obtains:
\[
\expv^{\sigma}(X_{n+1}\mid\mathcal{H}_n)=\mathbf{1}_{\theta_n=\bot}\cdot X_n+Y'_{\mathrm{p}}+Y'_\mathrm{a}+Y_\mathrm{c}+Y_\mathrm{d}
\]
(see below for details).
This can be seen as follows.
From the fact that $\mathbf{1}_{\theta_n=\bot}\cdot X_n$, $Y_\mathrm{d},Y_\mathrm{c}$ are measurable in $\mathcal{H}_n$,
we have $\expv^{\sigma}(\mathbf{1}_{\theta_n=\bot}\cdot X_n \mid  \mathcal{H}_n) = \mathbf{1}_{\theta_n=\bot}\cdot X_n$ and similarly for $Y_\mathrm{d},Y_\mathrm{c}$.
For $Y_\mathrm{p}$ and $Y_\mathrm{a}$ we need their conditional expectation as $Y'_{\mathrm{p}}$ and $Y'_{\mathrm{a}}$ defined below:
\[
Y'_\mathrm{p}:=\sum_{\loc\in \locs_\mathrm{p}}\left[\mathbf{1}_{\theta_n=\loc}\cdot\sum_{i\in \{0,1\}} \left(\probm^\sigma(Z_\loc=i)\cdot\eta(\loc_{Z_\loc=i},\overline{\vec{x}}_{n})\right)\right]
\]
and
\[
Y'_\mathrm{a}:= \sum_{\loc\in \locs_\mathrm{a}}\left[\mathbf{1}_{\theta_n=\loc}\cdot\expv_R\left(g_\loc,\overline{\vec{x}}_n\right)\right]
\]
where $g_\loc$ equals the function $(\vec{x},\vec{r})\mapsto\eta(\loc', f_\loc(\vec{x},\vec{r}))$.
Note that the fact that $\expv_R\left(g_\loc,\overline{\vec{x}}_n\right)$
is well-defined is because we consider polynomial functions (i.e., pRSMs).

Note that the case for $Y'_\mathrm{a}$ is derived from the fact that each $\eta(\loc',\cdot)$ is a polynomial over $X$ and $\overline{\vec{r}}_n$ is independent of $\mathcal{H}_n$.
Now by definition,
\[
\mathbf{1}_{\theta_n\not\in\locs_\mathrm{d}}\cdot\mathrm{pre}_\eta(\theta_n, \overline{\vec{x}}_n)=\mathbf{1}_{\theta_n=\bot}\cdot X_n+Y'_\mathrm{p}+Y'_\mathrm{a}+Y_\mathrm{c}
\]
and
\[
Y_\mathrm{d}\le \mathbf{1}_{\theta_n\in\locs_\mathrm{d}}\cdot\mathrm{pre}_\eta(\theta_n, \overline{\vec{x}}_{n})\enskip.
\]
Then the result follows.\qed
\end{proof}

\begin{remark}
In the proof of the above result, which generalizes the existing proof from LRSM to pRSM, 
the crucial property of pRSM we use is for assignments (locations in $L_a$) where we
used the well-definedness of $\expv_R\left(g_\loc,\overline{\vec{x}}_n\right)$ due to polynomials.
For more general RSMs if the well-definedness of $\expv_R\left(g_\loc,\overline{\vec{x}}_n\right)$ 
can be ensured then our proof ensures that the above result holds as well.
\end{remark}

To prove Theorem~\ref{thm:supermartingale-correctness}, one also needs an important property which states that an RSM falls below zero almost surely.

\begin{proposition}\label{prop:rsm}\cite{HolgerPOPL,DBLP:conf/popl/ChatterjeeFNH16}
Let $\{X_n\}_{n\in\Nset_0}$ be an RSM w.r.t a filtration $\{\mathcal{F}_n\}_{n\in\mathbb{N}_0}$ and constants $K,\epsilon$ (cf. Definition~\ref{def:rsm}).
Let $Z$ be the random variable defined by
$Z:=\min\{n\in\Nset_0\mid X_n<0\}$ with $\min\emptyset:=\infty$, denoting the first time $n$ that the RSM drops below~$0$.
Then $\probm(Z<\infty)=1$ and $\expv(Z)\le \frac{\expv(X_0)-K}{\epsilon}$~.
\end{proposition}

Now the proof for Theorem~\ref{thm:supermartingale-correctness} is as follows.

\textbf{Theorem~\ref{thm:supermartingale-correctness}.}
If there exists a $d$-pRSM $\eta$ w.r.t $(P,I)$ with constants $\epsilon,K$ (cf. Definition~\ref{def:prsm}), then $P$ is a.s. terminating and $\Eval(P)\le \UB(P):=\frac{\eta(\loc_0,\initval)-K}{\epsilon}$.
\begin{proof}
Let $\eta$ be a $d$-pRSM and $\{X_n\}_{n\in\Nset_0}$ be the stochastic process defined in Lemma~\ref{lemma:condexp}.
By Lemma~\ref{lemma:condexp}, C4 and the fact that $K\le -\epsilon$, $\{X_n\}_{n\in\Nset_0}$ is a ranking-supermartingale (w.r.t $\{\mathcal{H}_n\}_{n\in\Nset}$).
Then by C2, C3 and Proposition~\ref{prop:rsm},
\[
\Eval(P)=\sup_\sigma\expv^\sigma(T_P)\le\frac{\eta(\loc_0,\initval)-K}{\epsilon}\enskip.
\]\qed
\end{proof}

\section{Proof of Theorem~\ref{thm:concentration}}\label{app:concentration}

To prove Theorem~\ref{thm:concentration}, we need the following concentration inequality. 

\begin{theorem}[Hoeffding's Inequality~\cite{Hoeffding1963inequality,DBLP:conf/popl/ChatterjeeFNH16}]\label{thm:hoeffding}
Let $\{X_n\}_{n\in\mathbb{N}}$ be a supermartingale w.r.t some filtration $\{\mathcal{F}_n\}_{n\in\mathbb{N}}$ and $\{[a_n,b_n]\}_{n\in\mathbb{N}}$ be a sequence of intervals of positive length in $\mathbb{R}$.
If $X_1$ is a constant random variable and $X_{n+1}-X_n\in [a_n,b_n]$ a.s. for all $n\in\mathbb{N}$, then
\[
\mathbb{P}(X_n-X_1\ge\lambda)\le e^{-\frac{2\lambda^2}{\sum_{k=2}^n(b_k-a_k)^2}}
\]
for all $n\in\mathbb{N}$ and $\lambda> 0$.
\end{theorem}

Now we fix a difference-bounded $d$-pRSM $\eta$ w.r.t $[a,b]$. 
Recall that $X_n:=\eta(\theta_n, \overline{\vec{x}}_n)$. 
Define the stochastic process $\{Y_n\}_{n\in\mathbb{N}}$ by:
\[
Y_n=X_n+\epsilon\cdot(\min\{T_P,n\}-1)\enskip.
\]
The following proposition shows that $\{Y_n\}_{n\in\mathbb{N}}$ is a supermartingale and satisfies the requirements of Hoeffding's Inequaltiy.

\begin{proposition}\label{prop:hoeffding}
$\{Y_n\}_{n\in\mathbb{N}}$ is a supermartingale and $Y_{n+1}-Y_n\in [a+\epsilon,b+\epsilon]$ almost surely for all $n \in \Nats$.
\end{proposition}
\begin{proof}
Consider the following random variable:
\[
U_n=\min\{T_P, n+1\}-\min\{T_P,n\}\enskip,
\]
and observe that this is equal to $\mathbf{1}_{T_P>n}$.
From the properties of conditional expectation~\cite[Page 88]{probabilitycambridge} and the facts
that (i) the event $T_P>n$ is measurable in $\mathcal{F}_n$
(which implies that $\expv(\mathbf{1}_{T_P>n}\mid \mathcal{F}_n)= \mathbf{1}_{T_P>n}$);
and (ii) $X_n\ge 0$ iff $T_P>n$ (cf. conditions C2 and C3),
we have
\begin{eqnarray*}
\expv(Y_{n+1}\mid\mathcal{F}_n)-Y_n 
&=&\expv(X_{n+1}\mid\mathcal{F}_n)-X_n 
+\epsilon\cdot\expv(U_n \mid\mathcal{F}_n) \\
&=&\expv(X_{n+1}\mid\mathcal{F}_n)-X_n+\epsilon\cdot\expv(\mathbf{1}_{T_P>n}\mid\mathcal{F}_n)\\
&=&\expv(X_{n+1}\mid\mathcal{F}_n)-X_n+\epsilon\cdot\mathbf{1}_{T_P>n} \\
&\le &-\epsilon\cdot\mathbf{1}_{X_n\ge 0}+\epsilon\cdot\mathbf{1}_{T_P>n}\\
&=& 0\enskip.
\end{eqnarray*}
Note that the inequality above is due to the fact that $X_n$ is a ranking supermartingale.
Moreover, since $T_P\le n$ implies $\theta_n=\loc_P^\lout$ and $X_{n+1}=X_n$ we have that
$(X_{n+1}-X_n)= \mathbf{1}_{T_P>n}\cdot (X_{n+1}-X_n)$.
Hence we have
\begin{eqnarray*}
Y_{n+1}-Y_n &=& X_{n+1}-X_n+\epsilon\cdot U_n \\
&=& (X_{n+1}-X_n) +\epsilon\cdot \mathbf{1}_{T_P>n}   \\
&=& \mathbf{1}_{T_P>n}\cdot (X_{n+1}-X_n+\epsilon)\enskip.  \\
\end{eqnarray*}
Hence $Y_{n+1}-Y_n\in [a+\epsilon,b+\epsilon]$.\qed
\end{proof}

\begin{proof}[of Theorem~\ref{thm:concentration}]
Let $W_0:=Y_1=\eta(\loc_0,\initval)$.
Fix any demonic strategy $\sigma$.
By Hoeffing's Inequality, for all $\lambda>0$, we have 
$\mathbb{P}(Y_n-W_0\ge\lambda)\le e^{-\frac{2\lambda^2}{(n-1)(b-a)^2}}$.
Note that $T_P> n$ iff $X_n\ge 0$ by conditions C2 and C3 of pRSM.
Let $\alpha=\epsilon(n-1)-W_0$ and $\wh{\alpha}=\epsilon(\min\{n,T_P\}-1)-W_0$.
Note that with the conjunct $T_P>n$ we have that $\alpha$ and $\wh{\alpha}$
coincide.
Thus, for $\mathbb{P}(T_P > n)=\mathbb{P}(X_n\ge 0\wedge T_P>n)$ we have
\begin{eqnarray*}
\mathbb{P}(X_n\ge 0\wedge T_P>n) 
&=&\mathbb{P}((X_n+\alpha \ge \alpha ) \wedge (T_P>n))\\
&=&\mathbb{P}((X_n+\wh{\alpha} \ge \alpha ) \wedge (T_P>n))\\
&\le &\mathbb{P}((X_n+ \wh{\alpha} \geq \alpha))\\
&=&\mathbb{P}(Y_n-Y_1\ge \epsilon(n-1)-W_0)\\
&\le &e^{-\frac{2(\epsilon(n-1)-W_0)^2}{(n-1)(b-a)^2}}
\end{eqnarray*}
for all $n>\frac{W_0}{\epsilon}+1$.
The first equality is obtained by simply adding $\alpha$ on both
sides, and the second equality uses that because of the conjunct $T_P>n$
we have $\min\{n,T_P\}=n$ which ensures $\alpha=\wh{\alpha}$.
The first inequality is obtained by simply dropping the conjunct $T_P>n$.
The following equality is by definition, and the final inequality is an application of
Hoeffding's Inequality.\qed
\end{proof}

\begin{remark}
The above result holds for general difference-bounded RSMs and does not rely on 
the fact that it is a pRSM.
\end{remark}

\section{Proof for Theorem~\ref{thm:prsmdecision}}\label{app:thmprsmdecision}

\noindent\textbf{Theorem~\ref{thm:prsmdecision}}. 
The problem whether a (difference-bounded) $d$-pRSM w.r.t $(P,I)$ exists is decidable.
\begin{proof}
Let $\mbox{\sl M}$ be the set of all monomials of degree no greater than $d$.
Let a template for a $d$-pRSM be $\sum_{h\in M}a_h\cdot h$, where $a_h$ are scalar variables to be resolved.
Then it is straightforward that conditions C1-C4 can be directly encoded as formulae in the first-order theory of reals which is first existentially quantified over the variables $a_h, K,\epsilon$ and then universally quantified over the vector variable $\vec{x}$.
The conditions for difference-bounded pRSMs can also be encoded as formulae which are firstly existentially quantified over the scalar variables $a,b$ and then universally quantified over vector variable $\vec{x}$. 
Thus, the existence a (difference-bounded) $d$-pRSM is reduced to the validity of a formula in the first-order theory of reals, which is decidable ~\cite{Tarski1951,BasuPollackRoy2006}.\qed
\end{proof}

\section{Proof of Theorem~\ref{thm:synthesissoundness}}
\label{app:thmsynthesissoundness}

\noindent\textbf{Theorem~\ref{thm:synthesissoundness}.}
Any function $\eta$ synthesized through the algorithm \PRSMSynth{} is a valid pRSM.   
\begin{proof}
To prove the soundness we observe that Steps~1-3 of the algorithm are 
basically instantiation of the template and obtaining the coefficients.
Step~4 is the pre-expectation computation based on the definition.
The crucial step is Step~5 and Step~6.
The soundness of Step~5 and Step~6 follows from the soundness of 
Positivstellensatz's (cf. Theorem~\ref{thm:schmuedgen} to~\ref{thm:handelman}) 
regardless of the compactness of  $\SAT{\Gamma}$: either Eq.~($\ddag$), Eq.~($\S$) or Eq.~($\#$) guarantees that formula $(\dag)$ holds with 
`$g(\vec{x}) > 0$' replaced  by `$g(\vec{x})\ge 0$'. 
It ensures that the synthesized pRSM is indeed a pRSM.
\qed
\end{proof}

\section{Experimental Details}\label{app:experiments}

In the following description of the programs, we use ``$a\le f\le b$'' for an abbreviation of ``$a\le f\wedge f\le b$'', and ``$(x_{n_1},\dots,x_{n_k})^\mathrm{T}:=(E_{n_1},\dots,E_{n_k})^\mathrm{T}$'' as a compact form for assignment ``$x_{n_1},\dots,x_{n_k}:=E_{n_1},\dots,E_{n_k}$''. 
We also use $\mathrm{UNIF}(a,b)$ to denote the uniform distribution on $[a,b]$.
Besides, the invariants are written in a bracketed fashion $[\dots]$ and are put directly after the labels they are attached to. 
In all our examples the invariants are straightforward to obtain directly from the program. 

\begin{example}[Logistic Map]\label{ex:logisticmap}
Consider the logistic-map example adopted in \cite{DBLP:conf/vmcai/Cousot05}. 
The program is depicted in Fig.~\ref{fig:logisticmap}.
\end{example}

\lstset{language=affprob}
\lstset{tabsize=3}
\newsavebox{\proglogisticmap}
\begin{lrbox}{\proglogisticmap}
\begin{lstlisting}[mathescape]         
[$0\le a\le 1\wedge 0\le x\le 1$]
while $0\le a\le 0.999\wedge 0.001\le x\le 1$ do
   [$0\le a\le 0.999\wedge 0.001\le x\le 1$] 
   $x:=a*x*(1-x)$
od
\end{lstlisting}
\end{lrbox}
\begin{figure}[t]
\centering
\usebox{\proglogisticmap}
\caption{Logistic Map}
\label{fig:logisticmap}
\vspace{-1em}
\end{figure}

\begin{example}[Decay]\label{ex:decay}
Consider a decay example in Fig.~\ref{fig:decay} which is a discretized randomized version of the system of differential equations $x'=-x+y,y'=-x-y$; 
the ODE describes the exponential decay of any initial value to the origin.  
\end{example}

\lstset{language=affprob}
\lstset{tabsize=3}
\newsavebox{\progdecay}
\begin{lrbox}{\progdecay}
\begin{lstlisting}[mathescape]

[$x^2+y^2\le 2$] 
while $0.1\le x^2+y^2\le 1$ do
     [$0.1\le x^2+y^2\le 1$]
     $\begin{pmatrix} x \\ y\end{pmatrix}:=\begin{pmatrix} \mathrm{UNIF}(0.98, 1)*x + 0.01*y \\ \mathrm{UNIF}(0.98, 1)*y-0.01*x \end{pmatrix}$ 
od
\end{lstlisting}
\end{lrbox}
\begin{figure}[h]
\centering
\usebox{\progdecay}
\caption{Decay}
\label{fig:decay}
\end{figure}

\begin{example}[Random Walk]\label{ex:randomwalk}
Consider a demonic random-walk example in Fig.~\ref{fig:rdwalk} which mimics a random walk within a bounded region; 
the region is defined through two non-linear parabola curves instead of linear constraints.
\end{example}

\lstset{language=affprob}
\lstset{tabsize=3}
\newsavebox{\progrdwalk}
\begin{lrbox}{\progrdwalk}
\begin{lstlisting}[mathescape]

[$x^2+y^2\le 2$]

while $x^2+y\le 1\wedge x^2-y\le 1$ do
    [$x^2+y\le 1\wedge x^2-y\le 1$]
    $\begin{pmatrix} x \\ y\end{pmatrix}:=\begin{pmatrix} x+\mathrm{UNIF}(-0.1,0.1) \\ y+\mathrm{UNIF}(-0.1,0.1) \end{pmatrix}$ 
od
\end{lstlisting}
\end{lrbox}
\begin{figure}[h]
\centering
\usebox{\progrdwalk}
\caption{Random Walk}
\label{fig:rdwalk}
\end{figure}

\begin{example}[Gambler's Ruin]\label{ex:gamblersruin}
Finally, we consider the gambler's ruin in Example~\ref{ex:runningexample} with invariants given in Example~\ref{ex:prsm}.
\end{example}

\begin{example}[Gambler's Ruin Variant]\label{ex:gamblersruin2}
Consider a variant of Example~\ref{ex:gamblersruin} depicted in Fig.~\ref{fig:gamblersruin2}. 
Note that this example is another affine program that also does not admit a linear ranking supermartingale. 
\end{example}

\begin{example}[Nested Loop]\label{ex:nestedloop}
Consider the example in Fig.~\ref{fig:nestedloop}.
The example is a nested loop with two independent loop-control variables.
\end{example}

\lstset{language=affprob}
\lstset{tabsize=3}
\newsavebox{\proggamblersruinb}
\begin{lrbox}{\proggamblersruinb}
\begin{lstlisting}[mathescape]

[$0.7\le x\le y+0.3$]
while $1\le x\le y$ do
    [$1\le x\le y$]
    if $\star$ do     
        [$1\le x\le y$]        
        $x:=x+\mathrm{UNIF}(-0.3,0.3)$ 
    else 
        [$1\le x\le y$] 
        if prob(0.5) do
           [$1\le x\le y$]
           $x:=x+0.1$
        else 
           [$1\le x\le y$]
           $x:=x-0.1$
        fi
    fi
od
\end{lstlisting}
\end{lrbox}
\begin{figure}[h]
\centering
\usebox{\proggamblersruinb}
\caption{Gambler's Ruin Variant}
\label{fig:gamblersruin2}
\end{figure}

\lstset{language=affprob}
\lstset{tabsize=3}
\newsavebox{\prognestedloop}
\begin{lrbox}{\prognestedloop}
\begin{lstlisting}[mathescape]
[$x\le m+0.2\wedge n\ge 0$]
while $x\le m$ do
    [$x\le m\wedge n\ge 0$]
    y:=0;
    [$x\le m\wedge y\le n+0.2 \wedge n\ge 0$]
    while $y\le n$ do
        [$x\le m\wedge y\le n \wedge n\ge 0$]       
        $y:=y+\mathrm{UNIF}(-0.1,0.2)$
    od;
    [$x\le m\wedge y\ge n \wedge n\ge 0$]
    $x:=x+\mathrm{UNIF}(-0.1,0.2)$
od
\end{lstlisting}
\end{lrbox}
\begin{figure}[h]
\centering
\usebox{\prognestedloop}
\caption{Nested Loop}
\label{fig:nestedloop}
\end{figure}

\begin{thebibliography}{10}
\providecommand{\url}[1]{\texttt{#1}}
\providecommand{\urlprefix}{URL }

\bibitem{sedumi}
{SeDuMi 1.3}. {http://sedumi.ie.lehigh.edu/} (2008)

\bibitem{cplex}
{IBM ILOG CPLEX Optimizer Interactive Optimizer Community Edition 12.6.3.0}.
  {http://www-01.ibm.com/software/integration/optimization/cplex-optimizer/}
  (2010)

\bibitem{sostools}
{SOSTOOLS v3.00}. {http://www.cds.caltech.edu/sostools/} (2013)

\bibitem{DBLP:journals/fac/BabicCHR13}
Babic, D., Cook, B., Hu, A.J., Rakamaric, Z.: Proving termination of nonlinear
  command sequences. Formal Asp. Comput.  25(3),  389--403 (2013)

\bibitem{BaierBook}
Baier, C., Katoen, J.P.: Principles of model checking. MIT Press (2008)

\bibitem{BasuPollackRoy2006}
Basu, S., Pollack, R., Roy, M.: Algorithms in Real Algebraic Geometry.
  Springer, 2nd edn. (2006)

\bibitem{Billingsley:book}
Billingsley, P.: {Probability and Measure}. Wiley, 3rd edn. (1995)

\bibitem{BG05}
Bournez, O., Garnier, F.: Proving positive almost-sure termination. In: Giesl,
  J. (ed.) Term Rewriting and Applications, 16th International Conference,
  {RTA} 2005, Nara, Japan, April 19-21, 2005, Proceedings. Lecture Notes in
  Computer Science, vol. 3467, pp. 323--337. Springer (2005),
  \url{http://dx.doi.org/10.1007/978-3-540-32033-3_24}

\bibitem{DBLP:conf/cav/BradleyMS05}
Bradley, A.R., Manna, Z., Sipma, H.B.: Linear ranking with reachability. In:
  Etessami, K., Rajamani, S.K. (eds.) Computer Aided Verification, 17th
  International Conference, {CAV} 2005, Edinburgh, Scotland, UK, July 6-10,
  2005, Proceedings. Lecture Notes in Computer Science, vol. 3576, pp.
  491--504. Springer (2005)

\bibitem{DBLP:conf/vmcai/BradleyMS05}
Bradley, A.R., Manna, Z., Sipma, H.B.: Termination of polynomial programs. In:
  Cousot  \cite{DBLP:conf/vmcai/2005}, pp. 113--129

\bibitem{SriramCAV}
Chakarov, A., Sankaranarayanan, S.: Probabilistic program analysis with
  martingales. In: Sharygina, N., Veith, H. (eds.) Computer Aided Verification
  - 25th International Conference, {CAV} 2013, Saint Petersburg, Russia, July
  13-19, 2013. Proceedings. Lecture Notes in Computer Science, vol. 8044, pp.
  511--526. Springer (2013)

\bibitem{CFK16}
Chatterjee, K., Fu, H., Goharshady, A.K.: Termination analysis of probabilistic
  programs through positivstellensatz's. In: CAV (2016)

\bibitem{DBLP:conf/popl/ChatterjeeFNH16}
Chatterjee, K., Fu, H., Novotn{\'{y}}, P., Hasheminezhad, R.: Algorithmic
  analysis of qualitative and quantitative termination problems for affine
  probabilistic programs. In: Bod{\'{\i}}k, R., Majumdar, R. (eds.) Proceedings
  of the 43rd Annual {ACM} {SIGPLAN-SIGACT} Symposium on Principles of
  Programming Languages, {POPL} 2016, St. Petersburg, FL, USA, January 20 - 22,
  2016. pp. 327--342. {ACM} (2016),
  \url{http://doi.acm.org/10.1145/2837614.2837639}

\bibitem{DBLP:conf/tacas/ColonS01}
Col{\'{o}}n, M., Sipma, H.: Synthesis of linear ranking functions. In:
  Margaria, T., Yi, W. (eds.) Tools and Algorithms for the Construction and
  Analysis of Systems, 7th International Conference, {TACAS} 2001 Held as Part
  of the Joint European Conferences on Theory and Practice of Software, {ETAPS}
  2001 Genova, Italy, April 2-6, 2001, Proceedings. Lecture Notes in Computer
  Science, vol. 2031, pp. 67--81. Springer (2001)

\bibitem{DBLP:conf/vmcai/Cousot05}
Cousot, P.: Proving program invariance and termination by parametric
  abstraction, {L}agrangian relaxation and semidefinite programming. In: Cousot
   \cite{DBLP:conf/vmcai/2005}, pp. 1--24

\bibitem{DBLP:conf/popl/CousotC77}
Cousot, P., Cousot, R.: Abstract interpretation: {A} unified lattice model for
  static analysis of programs by construction or approximation of fixpoints.
  In: Graham, R.M., Harrison, M.A., Sethi, R. (eds.) Conference Record of the
  Fourth {ACM} Symposium on Principles of Programming Languages, Los Angeles,
  California, USA, January 1977. pp. 238--252. {ACM} (1977)

\bibitem{DBLP:conf/vmcai/2005}
Cousot, R. (ed.): Verification, Model Checking, and Abstract Interpretation,
  6th International Conference, {VMCAI} 2005, Paris, France, January 17-19,
  2005, Proceedings, Lecture Notes in Computer Science, vol. 3385. Springer
  (2005)

\bibitem{RandBook2}
Dubhashi, D., Panconesi, A.: Concentration of Measure for the Analysis of
  Randomized Algorithms. Cambridge University Press, New York, NY, USA, 1st
  edn. (2009)

\bibitem{Durrett}
Durrett, R.: Probability: Theory and Examples (Second Edition). Duxbury Press
  (1996)

\bibitem{EGK12}
Esparza, J., Gaiser, A., Kiefer, S.: Proving termination of probabilistic
  programs using patterns. In: Madhusudan, P., Seshia, S.A. (eds.) Computer
  Aided Verification - 24th International Conference, {CAV} 2012, Berkeley, CA,
  USA, July 7-13, 2012 Proceedings. Lecture Notes in Computer Science, vol.
  7358, pp. 123--138. Springer (2012)

\bibitem{FarkasLemma}
Farkas, J.: A fourier-f\'{e}le mechanikai elv alkalmaz\'{a}sai ({H}ungarian).
  Mathematikai\'{e}s Term\'{e}szettudom\'{a}nyi \'{E}rtesit\"{o}  12,  457--472
  (1894)

\bibitem{HolgerPOPL}
Fioriti, L.M.F., Hermanns, H.: Probabilistic termination: Soundness,
  completeness, and compositionality. In: Rajamani, S.K., Walker, D. (eds.)
  Proceedings of the 42nd Annual {ACM} {SIGPLAN-SIGACT} Symposium on Principles
  of Programming Languages, {POPL} 2015, Mumbai, India, January 15-17, 2015.
  pp. 489--501. {ACM} (2015)

\bibitem{rwfloyd1967programs}
Floyd, R.W.: Assigning meanings to programs. Mathematical Aspects of Computer
  Science  19,  19--33 (1967)

\bibitem{Foster53}
Foster, F.G.: On the stochastic matrices associated with certain queuing
  processes. The Annals of Mathematical Statistics  24(3),  pp. 355--360 (1953)

\bibitem{SemidefiniteProgramming}
Gr\"{o}tschel, M., Lovasz, L., Schrijver, A.: {G}eometric {A}lgorithms and
  {C}ombinatorial {O}ptimization. Springer-Verlag Berlin Heidelberg (1993)

\bibitem{HandelmanTheorem}
Handelman, D.: Representing polynomials by positive linear functions on compact
  convex polyhedra. Pacific J. Math.  132,  35--62 (1988)

\bibitem{Hoeffding1963inequality}
Hoeffding, W.: Probability inequalities for sums of bounded random variables.
  Journal of the American Statistical Association  58(301),  13--30 (1963)

\bibitem{MatrixAnalysis}
Horn, R.A., Johnson, C.R.: {M}atrix {A}nalysis. Cambridge University Press, 2nd
  edn. (2013)

\bibitem{Howard}
Howard, H.: Dynamic Programming and {Markov} Processes. MIT Press (1960)

\bibitem{AlgebraHungerford}
Hungerford, T.W.: {A}lgebra. Springer (1974)

\bibitem{kaelbling1998planning}
Kaelbling, L.P., Littman, M.L., Cassandra, A.R.: Planning and acting in
  partially observable stochastic domains. Artificial intelligence  101(1),
  99--134 (1998)

\bibitem{LearningSurvey}
Kaelbling, L.P., Littman, M.L., Moore, A.W.: Reinforcement learning: A survey.
  Journal of Artificial Intelligence Research  4,  237--285 (1996)

\bibitem{Kemeny}
Kemeny, J., Snell, J., Knapp, A.: Denumerable {Markov} Chains. D. Van Nostrand
  Company (1966)

\bibitem{KGFP09}
Kress-Gazit, H., Fainekos, G.E., Pappas, G.J.: Temporal-logic-based reactive
  mission and motion planning. IEEE Transactions on Robotics  25(6),
  1370--1381 (2009)

\bibitem{prism}
Kwiatkowska, M.Z., Norman, G., Parker, D.: {PRISM} 4.0: Verification of
  probabilistic real-time systems. In: Gopalakrishnan, G., Qadeer, S. (eds.)
  Computer Aided Verification - 23rd International Conference, {CAV} 2011,
  Snowbird, UT, USA, July 14-20, 2011. Proceedings. Lecture Notes in Computer
  Science, vol. 6806, pp. 585--591. Springer (2011)

\bibitem{MM04}
McIver, A., Morgan, C.: Developing and reasoning about probabilistic programs
  in \emph{pGCL}. In: Cavalcanti, A., Sampaio, A., Woodcock, J. (eds.)
  Refinement Techniques in Software Engineering, First Pernambuco Summer School
  on Software Engineering, {PSSE} 2004, Recife, Brazil, November 23-December 5,
  2004, Revised Lectures. Lecture Notes in Computer Science, vol. 3167, pp.
  123--155. Springer (2004), \url{http://dx.doi.org/10.1007/11889229_4}

\bibitem{MM05}
McIver, A., Morgan, C.: Abstraction, Refinement and Proof for Probabilistic
  Systems. Monographs in Computer Science, Springer (2005)

\bibitem{DBLP:conf/sas/Monniaux01}
Monniaux, D.: An abstract analysis of the probabilistic termination of
  programs. In: Cousot, P. (ed.) Static Analysis, 8th International Symposium,
  {SAS} 2001, Paris, France, July 16-18, 2001, Proceedings. Lecture Notes in
  Computer Science, vol. 2126, pp. 111--126. Springer (2001),
  \url{http://dx.doi.org/10.1007/3-540-47764-0_7}

\bibitem{RandBook}
Motwani, R., Raghavan, P.: Randomized Algorithms. Cambridge University Press,
  New York, NY, USA (1995)

\bibitem{PazBook}
Paz, A.: Introduction to probabilistic automata (Computer science and applied
  mathematics). Academic Press (1971)

\bibitem{DBLP:conf/vmcai/PodelskiR04}
Podelski, A., Rybalchenko, A.: A complete method for the synthesis of linear
  ranking functions. In: Steffen, B., Levi, G. (eds.) Verification, Model
  Checking, and Abstract Interpretation, 5th International Conference, {VMCAI}
  2004, Venice, January 11-13, 2004, Proceedings. Lecture Notes in Computer
  Science, vol. 2937, pp. 239--251. Springer (2004)

\bibitem{PutinarPositivstellensatz}
Putinar, M.: Positive polynomials on compact semi-algebraic sets. Indiana Univ.
  Math. J.  42,  969--984 (1993)

\bibitem{Rabin63}
Rabin, M.: Probabilistic automata. Information and Control  6,  230--245 (1963)

\bibitem{SumitPLDI}
Sankaranarayanan, S., Chakarov, A., Gulwani, S.: Static analysis for
  probabilistic programs: inferring whole program properties from finitely many
  paths. In: PLDI. pp. 447--458 (2013)

\bibitem{ScheidererSurvey}
Scheiderer, C.: Positivity and sums of squares: A guide to recent results. The
  {IMA} Volumes in Mathematics and its Applications  149,  271--324 (2008)

\bibitem{SchmudgenPositivstellensatz}
Schm\"{u}dgen, K.: The ${K}$-moment problem for compact semi-algebraic sets.
  Math. Ann.  289,  203--206 (1991)

\bibitem{DBLP:books/daglib/0090562}
Schrijver, A.: {T}heory of {L}inear and {I}nteger {P}rogramming.
  Wiley-Interscience series in discrete mathematics and optimization, Wiley
  (1999)

\bibitem{DBLP:journals/jossac/ShenWYZ13}
Shen, L., Wu, M., Yang, Z., Zeng, Z.: Generating exact nonlinear ranking
  functions by symbolic-numeric hybrid method. J. Systems Science {\&}
  Complexity  26(2),  291--301 (2013)

\bibitem{DBLP:conf/pods/SohnG91}
Sohn, K., Gelder, A.V.: Termination detection in logic programs using argument
  sizes. In: Rosenkrantz, D.J. (ed.) Proceedings of the Tenth {ACM}
  {SIGACT-SIGMOD-SIGART} Symposium on Principles of Database Systems, May
  29-31, 1991, Denver, Colorado, {USA}. pp. 216--226. {ACM} Press (1991)

\bibitem{Tarski1951}
Tarski, A.: A decision method for elementary algebra and geometry. In:
  {Q}uantifier {E}limination and {C}ylindrical {A}lgebraic {D}ecomposition, pp.
  24--84. {T}exts and {M}onographs in {S}ymbolic {C}omputation, Springer Vienna
  (1951)

\bibitem{probabilitycambridge}
Williams, D.: {P}robability with {M}artingales. Cambridge University Press
  (1991)

\bibitem{DBLP:journals/fcsc/YangZZX10}
Yang, L., Zhou, C., Zhan, N., Xia, B.: Recent advances in program verification
  through computer algebra. Frontiers of Computer Science in China  4(1),
  1--16 (2010), \url{http://dx.doi.org/10.1007/s11704-009-0074-7}

\end{thebibliography}
\end{document}